%% file: arxiv.tex
\documentclass[authoryear,sort&compress,9pt,twocolumn,nocopyrightspace]{sigplanconf}

\input{definitions}
\usepackage{tikz-cd}

\newcommand{\ncsubl}{\not \hspace{-2mm}\csubl}
\newcommand{\Evidence}{\oblset{Evidence}}
\newcommand{\lobs}{\zeta}
\newcommand{\rel}{\approx_{\lobs}}
\newcommand{\rcomp}[1]{\mathcal{C}(#1)} 
\newcommand{\bval}[1]{bval(#1)}
\newcommand{\subst}{\sigma}
\newcommand{\itm}[1]{{t^{\cS_{#1}}}}
\renewcommand{\cast}[2]{#1#2}

\newcommand{\D}{\mathcal{D}}
\newcommand{\E}{\mathcal{E}}

\begin{document}

\title{Deriving a Simple Gradual Security Language}

\authorinfo{Ronald Garcia}{
Software Practices Lab \\
Department of Computer Science\\
University of British Columbia}
{rxg@cs.ubc.ca}
\authorinfo{\'Eric Tanter}
{PLEIAD Laboratory\\
Computer Science Department (DCC)\\
University of Chile}
{etanter@dcc.uchile.cl}

\maketitle

\begin{abstract}
  Abstracting Gradual Typing (AGT) is an approach to systematically deriving
  gradual counterparts to static type disciplines~\cite{garciaAl:popl2016}. The
  approach consists of defining the semantics of gradual types by interpreting
  them as sets of static types, and then defining an optimal abstraction back
  to gradual types. These operations are used to lift the static discipline to
  the gradual setting. The runtime semantics of the gradual language
  then arises as reductions on gradual typing derivations.

  To demonstrate the flexibility of AGT, we gradualize \linebreak
  \lsec~\cite{zdancewic},
  the prototypical security-typed language, with respect to only security
  labels rather than entire types, yielding a type system that ranges gradually
  from simply-typed to securely-typed. We establish noninterference for the
  gradual language, called \lgsec, using Zdancewic's logical relation proof
  method.  Whereas prior work presents gradual security \emph{cast} languages,
  which require explicit security casts, this work yields the first gradual
  security \emph{source} language, which requires no explicit casts.
\end{abstract}

\section{Introduction}

Gradual typing has often been viewed as a means to combine the agility benefits
of \emph{dynamic languages}, like Python and Ruby with the
reliability benefits of \emph{static languages} like OCaml and
Scala.  This paper, in a line of work on the foundations of gradual
typing, explores the idea that static and dynamic are merely relative notions.

This relativistic view of gradual typing is not new.  Work on gradual
information flow security by \citet{disney11flow} and
\citet{fennellThiemann:csf2013} develop languages where only information-flow
security properties are subject to a mix of dynamic and static checking.
\citet{banados14effects} develop a language where only computational effect
capabilities are gradualized.  In each of these cases, the ``fully-dynamic''
corner of the gradual language is not dynamic at all by typical standards, but
rather simply typed.  However, those languages support seamless migration
toward a \emph{more precise} typing discipline that subsumes simple typing.

To explore this notion, we revisit the idea of gradual information-flow
security.  Our tool of inquiry is a new approach to the foundations of gradual
typing called Abstracting Gradual Typing (AGT) \cite{garciaAl:popl2016}.  AGT
is a technique for systematically deriving gradual type systems by interpreting
gradual types as sets of static types.  That work developed a traditional
gradual type system with subtyping, introducing an \emph{unknown type} \?. But
AGT was directly inspired by~\citet{banados14effects}, who used an early
version of these techniques to gradualize only effects. However, they develop
the dynamic semantics of gradual effects in the traditional ad hoc fashion.
Here, we bring the approach full circle, deriving a complete static and dynamic
semantics for a gradual counterpart to the \lsec language of \citet{zdancewic}.

In their simplest form, security-typed languages require values and types to be
annotated with security labels, indicating their confidentiality level.  The
security type system guarantees \emph{noninterference}, \ie, that
more-confidential information does not alter the less-confidential results of
any expression.

We prove that the resulting gradual language, called \lgsec, is not only
\emph{safe} in that it never unexpectedly crashes, but that it is \emph{sound}
in that it honours the information-flow invariants of the precisely typed
terms.  The former property is unsurprising, since even the most imprecisely
typed program still maintains the simple typing discipline, which is enough to
establish the safety of the operational semantics.  The soundness of the
language with respect to the security type discipline, \ie, that basic
information flow properties are respected, is the key property.

The prior work in gradual security typing developed gradual \emph{cast
  languages}, which require explicit type casts to connect imprecisely typed
terms with precisely typed terms.  This is akin to the intermediate languages
of traditional gradually-typed languages.  This work presents the first gradual
\emph{source} language, where no explicit casts are needed: they are introduced
by the language semantics.  Furthermore, following the AGT approach, the
runtime semantics are induced by the proof of type safety for \lsec, yielding a
crisp connection to that precise static type discipline.

As with the original work on AGT, we can straightforwardly establish proper
adaptations of the refined criteria for gradually typed languages.  We will do
so in this ongoing work.

Ultimately this work views gradual typing as a theory of \emph{imprecise
  typing} rather than dynamic checking.  Indeed dynamic checking is an
inevitable consequence of this approach, but the focus here is on the types and
their meaning.  We believe that this broader view of gradual typing can widen
the reach of gradual typing beyond its current niche of interest among
dynamic language enthusiasts.  Furthermore, we believe that AGT generalizes the
foundations of gradual typing enough to support a wide variety of
gradual type disciplines.

\section{The Static Language: \lsec}

\begin{figure}[h]
  \begin{small}
  \begin{displaymath}
    \begin{array}{rcll}
      \multicolumn{4}{c}{
        \lx \in \Label,\quad
        S \in \Type,\quad
        x \in \oblset{Var},\quad
        b \in \oblset{Bool},\quad
        \oplus \in \oblset{BoolOp}
      }\\
      \multicolumn{4}{c}{
        t \in \Term,\quad
      
        r \in \oblset{RawValue}\quad      
        v \in \oblset{Value}\quad
        \Gamma \in \oblset{Var} \finto \Type
      } \\[3mm]
      S & ::= & \Bool_{\lx} | S ->_{\lx} S & \text{(types)} \\
      b & ::= & \ttt | \ff & \text{(Booleans)}\\
      r & ::= & b | \lambda x:S.t & \text{(raw values)}\\
      v & ::= & r_\lx & \text{(values)}\\
      t & ::= & v |  t\;t | t \oplus t | \ite{t}{t}{t} | t :: S 
      & \text{(terms)} \\
      \oplus & ::= & \land | \lor | \implies & \text{(operations)} \\
    \end{array}
  \end{displaymath}

\framebox{$\Gamma |- t : S$}
 \begin{mathpar}
    \inference[(Sx)]{x:S\in\Gamma}{\Gamma |- x : S}
    \and
    \inference[(Sb)]{}{\Gamma |- b_\lx : \Bool_\lx}
    \and
    \inference[(S$\lambda$)]{
      \Gamma, x:S_1 |- t : S_2
    }{
      \Gamma |- (\lambda x : S_1.t)_\lx : S_1 ->_{\lx} S_2
    }
    \and
    \inference[(S$\oplus$)]{
      \Gamma |- t_1 : \Bool_{\lx_1} & 
      \Gamma |- t_2 : \Bool_{\lx_2} &
    }{
      \Gamma |- t_1 \oplus t_2 : \Bool_{(\lx_1 \ljoin \lx_2)}
    }
    \and
    \inference[(Sapp)]{
      \Gamma |- t_1 : S_{11} ->_\lx S_{12} &
      \Gamma |- t_2 : S_2 & S_2 \sub S_{11}
    }{
      \Gamma |- t_1\;t_2 : S_{12} \ljoin \lx
    }
    \and
    \inference[(Sif)]{
      \Gamma |- t : \Bool_\lx &
      \Gamma |- t_1 : S_1 & \Gamma |- t_2 : S_2
    }{
      \Gamma |- \ite{t}{t_1}{t_2} : (S_1 \subjoin S_2) \ljoin \lx
    }
    \and
    \inference[(S::)]{
      \Gamma |- t : S_1 &
      S_1 <: S_2
    }{
      \Gamma |- t :: S_2 : S_2
    }
  \end{mathpar}

\framebox{$S \sub S$}
 \begin{mathpar} 
   \inference{
     \lx \subl \lx'
   }{
     \Bool_\lx \sub \Bool_{\lx'}
   }
   \and
   \inference{
     S'_1 \sub S_1 &
     S_2 \sub S'_2 &
     \lx \subl \lx'
   }{
     S_1 ->_{\lx} S_2 \sub S'_1 ->_{\lx'} S'_2
   }
  \end{mathpar}

\framebox{$S \subjoin S$, $S \submeet S$}
\begin{equation*}
  \begin{block}
    \subjoin : \Type \times \Type \rightharpoonup \Type \\[0.2em]
    \Bool_\lx \subjoin \Bool_{\lx'} = \Bool_{(\lx \ljoin \lx')} \\
    (S_{11} ->_\lx S_{12}) \subjoin (S_{21} ->_{\lx'} S_{22}) = 
    (S_{11} \submeet S_{21}) ->_{(\lx \ljoin \lx')} (S_{12} \subjoin S_{22}) \\
    S \subjoin S \text{ undefined otherwise} \\
    \\
    \submeet : \Type \times \Type \rightharpoonup \Type \\[0.2em]
    \Bool_\lx \submeet \Bool_{\lx'} = \Bool_{(\lx \lmeet \lx')} \\
    (S_{11} ->_{\lx} S_{12}) \submeet (S_{21} ->_{\lx'} S_{22}) = 
    (S_{11} \subjoin S_{21}) ->_{(\lx \lmeet \lx')} (S_{12} \submeet S_{22}) \\
    S \submeet S \text{ undefined otherwise} \\
  \end{block}
\end{equation*}

\end{small}
 \caption{\lsec: Syntax and Static Semantics}
  \label{fig:lsec}
\end{figure}

We first present the \lsec language, with some differences from the original
presentation~\cite{zdancewic}. The most notable changes are that the type
system is syntax directed, and the runtime semantics are small-step structural
operational semantics rather than big-step natural semantics.

Figure~\ref{fig:lsec} presents the syntax and type system for \lsec.  The
language extends a simple typing discipline with a lattice of \emph{security
  labels} $\lx$.  All program values are ascribed security labels, which are
partial ordered $\subl$ from low security to high-security and include
top and bottom security labels $\top$ and $\bot$.  The \lsec types $S$ extend a
simple type discipline by associating a security label to each type constructor.

Rules for variables, constants and functions are straightforward.  The
(S$\oplus$) rule for binary boolean operations ensures that the confidentiality
of combining two values is the least upper bound, or \emph{join} $\ljoin$, of
the confidentiality of the two sub-expressions. Similarly, when applying a
function (Sapp), the result type joins the label of the function's result type
$S_{12}$ with the label $\lx$ of the function type. For this, the rule uses a
notion of \emph{label stamping} on types:%
\footnote{We overload the join notation $\ljoin$ throughout, and rely on the
  context to disambiguate.}
\begin{displaymath}
\begin{array}{c}
\Bool_{\lx} \ljoin \lx' = \Bool_{(\lx \ljoin {\lx'})} \\
(S_1 ->_{\lx} S_2) \ljoin \lx' = S_1 ->_{(\lx \ljoin {\lx'})} S_2 \\
\end{array}
\end{displaymath}

Rule (Sapp) also appeals to a notion of subtyping $S \sub S$.  Subtyping is
induced by the ordering on security labels.  It allows lower-security values to
flow to higher-security contexts, but not vice-versa.  Rule (Sif) specifies
that the type of a conditional is the
\emph{subtyping join}
$\subjoin$ of the types of the branches, further stamped to incorporate the
confidentiality of the predicate expression's label $l$.  The latter is
necessary to forbid indirect flow of information through the conditional.  As
usual, the join of two function types is defined in terms of the \emph{meet}
$\submeet$ of the argument types, which in turn relies on the \emph{label meet}
operator $\lmeet$.
The (S::) rule introduces ascription, which can move the type of an expression
to any supertype.

These syntax-directed typing rules define a type system that is sound and
complete with respect to Zdancewic's.  The following propositions use $|-^{Z}$
for the type system of~\citet{zdancewic}, and consider only terms without
ascription (\ie, the common subset of the two systems).

\begin{restatable}{proposition}{syntypein}
  If $\Gamma |- t : S$ then $\Gamma |-^{Z} t : S$.
\end{restatable}
\begin{proof}
  By induction on $\Gamma |- t : S$.  
\end{proof}

\begin{restatable}{proposition}{insyntype}
  If $\Gamma |-^{Z} t : S$ Then $\Gamma |- t : S'$ for some $S' <: S$.
\end{restatable}
\begin{proof}
  By induction on $\Gamma |-^{Z} t : S$.  
\end{proof}

\begin{figure}[t]
\begin{small}
\begin{equation*}
  \begin{array}{rcll}
    t & ::= & \ldots | t \ljoin \lx & \text{(term stamping)}\\
    f & ::= &  [] \oplus t | v \oplus [] | []\;t | v\;[] & \text{(frames)}\\
      & & \ite{[]}{t}{t} | [] \ljoin \lx & 
  \end{array}
\end{equation*}

\begin{mathpar}
  \inference[(S$\ljoin$)]{
    \Gamma |- t : S
  }{
    \Gamma |- t \ljoin \lx : S \ljoin \lx
  }
\end{mathpar}

\framebox{$t --> t$}
\noindent\textbf{Notions of Reduction}
\begin{equation*}
\begin{block}
  \mbox{$b_1$}_{\lx_1} \oplus \mbox{$b_2$}_{\lx_2} --> 
(b_1 \;\llbracket \oplus \rrbracket\; b_2)_{(\lx_1 \ljoin \lx_2)} 
\\[0.8em]
(\lambda x:S.t)_\lx \; v --> ([v/x]t) \ljoin \lx 
\\[0.8em]
\ite{\ttt_\lx}{t_1}{t_2} --> t_1 \ljoin \lx
\\[0.8em]
\ite{\ff_\lx}{t_1}{t_2} --> t_2 \ljoin \lx
\\[0.8em]
r_{\lx_1} \ljoin \lx_2 --> r_{(\lx_1 \ljoin \lx_2)}
\end{block}
\end{equation*}
\framebox{$t \red t$}
\textbf{Reduction}
\begin{mathpar}
  \inference[]{
    t_1 --> t_2
  }{
    t_1 \red t_2
  } 
  \and
  \inference[]{
    t_1 \red t_2
  }{
    f[t_1] \red f[t_2]
  } 
\end{mathpar}
\end{small}
 \caption{\lsec: Small-Step Dynamic Semantics}
  \label{fig:lsec-dyn}
\end{figure}

\paragraph{Dynamic semantics.}
The dynamic semantics of \lsec were originally presented as big step
semantics~\cite{zdancewic}. Figure~\ref{fig:lsec-dyn} presents the equivalent
small-step semantics. Of particular interest is the new label stamping form on
terms, which we call \emph{term stamping} $t \ljoin \lx$.  Term stamping allows
small-step reduction to retain security information that is merged with the
resulting value of the nested term.

This small-step semantics coincides with the big-step semantics of
\lsec~\cite{zdancewic}. Note that we establish the equivalence to the source
\lsec language (Figure~\ref{fig:lsec}), \ie, without term stamping, since it is
only needed internally to support small-step reduction.
As usual, $\red^{*}$ denotes the reflexive, transitive closure of $\red$.

\begin{restatable}[]{proposition}{bsss}
$t \Downarrow v$ if and only if $t \red^{*} v$.
\end{restatable}
\begin{proof}
  \mbox{}
  \begin{case}[only if]
    By induction on $t \Downarrow v$, using the
    admissibility of the
    $\Downarrow$ rules in $\red^{*}$.\\
  \end{case}

  \begin{case}[if]
    By induction on the length of the reduction $t \red^{*} v$. Straightforward
    case analysis on $t$ using the admissibility of the inversion lemmas for
    $\Downarrow$ in $\red^{*}$.
  \end{case}
\end{proof}

\section{Gradualizing \lsec}
\label{sec:gradualizing-lsec-1}

In gradualizing \lsec, we could decide to support unknown information in both
types and security labels. Here, to show the flexibility of the AGT approach,
we gradualize \lsec only in terms of security labels, thereby supporting a
gradual evolution between simply-typed programs and securely-typed programs.

\subsection{Meaning of Unknown Security Type Information}
\label{ssec:sec-meaning}

To gradualize our security types, we introduce a notion of \emph{gradual
  labels} and define their meanings in terms of concrete labels of a given
security lattice.

\begin{definition}[Gradual labels]
A gradual label $\clx$ is either a label $\lx$ or the unknown label $\ul$.
\end{definition}
\begin{displaymath}
\begin{array}{rcll}
\multicolumn{4}{l}{\clx \in \GLabel}\\
  \clx & ::= & \lx \;|\; \ul & \text{(gradual labels)}
\end{array}
\end{displaymath}

As with static security typing, we develop gradual security types by assigning
a gradual label to every type constructor.

\begin{definition}[Gradual security type]
A gradual security type is a gradual type labeled with a gradual label:
\end{definition}
\begin{displaymath}
\begin{array}{rcll}
\multicolumn{4}{l}{\cS \in \GType}\\
\cS & ::= & \Bool_{\clx} | \cS ->_{\clx} \cS & \text{(gradual types)}\\
\end{array}
\end{displaymath}

To give meaning to gradual security types, we use the AGT approach of defining
a concretization function that maps gradual security types to sets of static
security types.  This concretization is the natural lifting of a concretization
for gradual labels.

\begin{definition}[Label Concretization]
Let ${\gammal : \GLabel -> \Pow(\Label)}$ be defined as follows:
\begin{align*}
  \gammal(\lx) &= \set{\lx} \\
  \gammal(\ul) &= \Label
\end{align*}
\end{definition}

\noindent We give meaning to the unknown label by saying that it represents 
\emph{any label}.  On the other hand, any static label represents only itself.

Since we are operating on complete lattices, the sound and optimal abstraction
function from sets of labels to gradual labels is fully determined by the
concretization.  We characterize it below.

\begin{definition}[Label Abstraction]
  Let 
  $\alphal : \Pow(\Label) -> \GLabel$
  be defined as follows:
  \begin{align*}
    \alphal(\set{\lx})  &= \lx \\
    \alphal(\emptyset)  &\text{ is undefined} \\
    \alphal(\cll)  &= \;\ul \text{  otherwise}
  \end{align*}
\end{definition}

\begin{restatable}[$\alphal$ is Sound]{proposition}{laisound}
  \label{prop:lai-sound}
  If $\cll$ is not empty, then 
  ${\cll \subseteq \gammal(\alphal(\cll))}$.
\end{restatable}
\begin{proof}
By case analysis on the structure of $\cll$.  
If $\cll = \set{\lx}$ then $\gammal(\alphal(\set{\lx})) =
\gammal(\lx) = \set{\lx} = \cll$,
otherwise \\
${\gammal(\alphal(\cll)) = \gammal(\ul) = \Label \supseteq \cll}$.
\end{proof}

\begin{restatable}[$\alphal$ is Optimal]{proposition}{laiopt}
  \label{prop:lai-opt}
  If $\cll \subseteq \gammal(\clx)$ then 
  $\alphal(\cll) \sqsubseteq \clx$.
\end{restatable}
\begin{proof}
By case analysis on the structure of $\clx$.  
If $\clx = \lx$, $\gammal(\clx)= \set{\lx}$;
$\cll \subseteq \set{\lx}, \cll \neq \emptyset$ implies $\alphal(\cll) =
\alphal(\set{\lx}) = \lx \sqsubseteq \clx$ (if $\cll = \emptyset$, 
$\alphal(\cll)$ is undefined).
If $\clx = \ul$, $\clx' \sqsubseteq \clx$ for all $\clx'$.
\end{proof}

Having defined the meaning of gradual labels, we define the meaning of
gradual security types via concretization.

\newsavebox{\grr}
\begin{definition}[Type Concretization]
\mbox{}
Let ${\gammas : \GType -> \Pow(\Type)}$ be defined as follows:
\begin{align*}
\gammas(\Bool_{\clx}) &= \set{\Bool_{\lx} | \lx \in \gammal(\clx)} \\
\gammas(\cS ->_{\clx} \cS) &=
\savebox{\grr}{$\gammas(\cS_1) ->_{\gammal(\clx)} \gammas(\cS_2)$}
\wideparen{\usebox{\grr}}
\intertext{where}
\savebox{\grr}{$\clS_1 ->_{\cll} \clS_2$}
\wideparen{\usebox{\grr}} &= 
\set{S_1 ->_{\lx} S_2 | S_i \in \gammas(\clS_i), \lx \in \cll}.
\end{align*}
\end{definition}

With concretization of security type, we can now define security type
precision.

\begin{definition}[Label and Type Precision]
  \mbox{}
  \begin{enumerate}
  \item $\clx_1$ is less imprecise than $\clx_2$, notation
    $\clx_1 \gprec \clx_2$, if and only if
    $\gammal(\clx_1) \subseteq \gammal(\clx_2)$; inductively:
\begin{mathpar} 
  \inference{}{\clx \gprec \?}
  \and
  \inference{}{\clx \gprec \clx}
  \and 
\end{mathpar}

  \item $\cS_1$ is less imprecise than $\cS_2$,
    notation $\cS_1 \sqsubseteq \cS_2$, if and only if
    $\gammas(\cS_1) \subseteq \gammas(\cS_2)$; inductively:
\begin{mathpar}
  \inference{
    \clx_1 \gprec \clx_2
  }{
    \Bool_{\clx_1} \gprec \Bool_{\clx_2}
  }
  \and
  \inference{
    \cS_{11} \gprec \cS_{21} & \cS_{12} \gprec \cS_{22} \\ \clx_1 \gprec \clx_2
  }{
    \cS_{11} ->_{\clx_1} \cS_{12} \gprec \cS_{21} ->_{\clx_2} \cS_{22}
  }
\end{mathpar}
  \end{enumerate}
\end{definition}

We now define the abstraction function.

\begin{definition}[Type Abstraction]
  Let the abstraction function 
  $\alphas : \Pow(\Type) -> \GType$
  be defined as:
\begin{align*}
  \alphas(\set{\overline{\Bool_{\lx_i}}}) &=
  \Bool_{\alphal(\set{\overline{\lx_i}})} \\[1mm]
  \alphas(\set{\overline{S_{i1} ->_{\lx_i} S_{i2}}}) &=
  \alphas(\set{\overline{S_{i1}}}) ->_{\alphal(\set{\overline{\lx_i}})} 
  \alphas(\set{\overline{S_{i2}}}) \\[1mm]
  \alphas(\clT) & \text{ is undefined otherwise}\\[1mm]
  \end{align*}
\end{definition}

We can only abstract {\em valid} sets of security types, \ie~in which
elements only defer by security labels.

\begin{definition}[Valid Type Sets]
\begin{mathpar} 
  \inference{}{\valid(\setof{\Bool_{\lx_i}})}

  \inference{\valid(\setof{S_{i1}}) & \valid(\setof{S_{i2}})}
  {\valid(\setof{(S_{i1} -> S_{i2})_{\lx_i}})}
  \end{mathpar}
\end{definition}

\begin{restatable}[$\alphas$ is Sound]{proposition}{saisound}
  \label{prop:sai-sound}
  \mbox{}\\
  If $\valid(\clS)$ then 
  $\clS \subseteq \gammas(\alphas(\clS)).$
\end{restatable}
\begin{proof}
  By well-founded induction on $\clS$ according
  to the ordering relation $\clS \sqsubset \clS$ defined as follows:
  \begin{align*}
    \cldom(\clS) \sqsubset \clS\\
    \clcod(\clS) \sqsubset \clS
  \end{align*}
  Where $\cldom,\clcod : \Pow(\GType) -> \Pow(GType)$ are the collecting
  liftings of the domain and codomain functions $\dom,\cod$ respectively, \eg,
  \begin{equation*}
    \cldom(\clS) = \set{\dom(S) | S \in \clS}.
  \end{equation*}

  We then consider cases on $\clS$ according to the definition of $\alphas$.

\begin{case}[$\set{\overline{\Bool_{\lx_i}}}$]
\begin{align*}
\gammas(\alphas(\set{\overline{\Bool_{\lx_i}}})) &=  
\gammas(\Bool_{\alphal(\set{\overline{\lx_i}})}) \\
&= \set{\Bool_{\lx} | \lx \in \gammal(\alphal(\set{\overline{\lx_i}}))} \\
                       &\supseteq \set{\overline{\Bool_{\lx_i}}}
                       \text{ by soundness of $\alphal$.}
\end{align*}
\end{case}

\begin{case}[$\set{\overline{S_{i1} ->_{\lx_i} S_{i2}}}$]
\begin{align*} 
& \gammas(\alphas(\set{\overline{S_{i1} ->_{\lx_i} S_{i2}}})) \\
=\; &
\gammas(\alphas(\set{\overline{S_{i1}}}) ->_{\alphal(\set{\overline{\lx_i}})} 
  \alphas(\set{\overline{S_{i2}}})) \\
=\; &
\gammas(\alphas(\set{\overline{S_{i1}}}) ->_{\gammal(\alphal(\set{\overline{\lx_i}}))}
\gammas(\alphas(\set{\overline{S_{i2}}}))
\\
\supseteq\;& \set{\overline{S_{i1} ->_{\lx_i} S_{i2}}}
\end{align*}
by induction hypotheses on $\set{\overline{S_{i1}}}$ and
  $\set{\overline{S_{i2}}}$, and soundness of $\alphal$.
\end{case}
\end{proof}

\begin{restatable}[$\alphas$ is Optimal]{proposition}{saiopt}
  \label{prop:sai-opt}
  If $\valid(\clS)$ and $\clS \subseteq \gammas(\cS)$ then 
  $\alphas(\clS) \sqsubseteq \cS$.
\end{restatable}
\begin{proof}
By induction on the structure of $\cS$.
\begin{case}[$\Bool_{\clx}$] 
  $\gammas(\Bool_\clx) = \set{\Bool_{\lx} | \lx \in {\gammal(\clx)}}$\\
  So $\clS = \set{\Bool_{\lx} | \lx \in \cll}$ for some
  $\cll \subseteq \gammal(\clx)$. By optimality of $\alphal$,
  $\alphal(\cll) \gprec \clx$, so
  $\alphas(\set{\Bool_{\lx} | \lx \in \cll}) = \Bool_{\alphal(\cll)}
  \gprec \Bool_\clx$.
\end{case}

\begin{case}[$\cS_1 ->_{\clx} \cS_2$]
  $\gammas(\cS_1 ->_\clx \cS_2) =
  \savebox{\grr}{$\gammas(\cS_1) ->_{\gammal(\clx)} \gammas(\cS_2)$}
  \wideparen{\usebox{\grr}}.$
  \\
  So 
  $\clS = \set{\overline{S_{1i} ->_{l_i} S_{2i}}}$,
  with $\set{\overline{S_{1i}}} \subseteq \gammas(\cS_1)$,\\
  $\set{\overline{S_{1i}}} \subseteq \gammas(\cS_2)$, and
  $\set{\overline{l_i}} \subseteq \gammal(\clx)$.
  By induction hypotheses,
  $\alphas(\set{\overline{S_{1i}}}) \gprec \cS_1$ and 
  $\alphas(\set{\overline{S_{2i}}}) \gprec \cS_2$, and by
  optimality of $\alphal$, $\alphal(\set{\overline{l_i}}) \gprec \clx$.
  Hence
  $\alphas(\set{\overline{S_{1i} ->_{l_i} S_{2i}}})
  =\\
  \alphas(\set{\overline{S_{1i}}}) ->_{\alphal(\set{\overline{l_i}})}
  \alphas(\set{\overline{S_{2i}}})
  \gprec \cS_1 ->_\clx \cS_2$.
\end{case}

\end{proof}

\subsection{Consistent Predicates and Operators}

Following the AGT approach, we lift predicates on labels and types to
\emph{consistent} predicates on the corresponding gradual labels and gradual
types.  Consistent predicates hold if \emph{some} member of the collecting
semantics satisfies the corresponding static predicate.  We lift partial
functions to gradual partial functions, as per the standard approach in
abstract interpretation.

\begin{definition}[Consistent label ordering]
  $\clx_1 \csubl \clx_2$ if and only if $\lx_1 \subl \lx_2$ 
  for some $(\lx_1, \lx_2) \in \gammal(\clx_1) \times \gammal(\clx_2).$
\\
\noindent Algorithmically:
 \begin{mathpar}
   \inference[]{}{\? \csubl \clx}
   \and
   \inference[]{}{\clx \csubl \?}
   \and
   \inference[]{\lx_1 \subl \lx_2}{\lx_1 \csubl \lx_2}
 \end{mathpar}
\end{definition}

 \begin{definition}[Gradual label join]
$\\ \clx_1 \cjoin \clx_2 = \alphal(\set{ \lx_1 \ljoin \lx_2 | (\lx_1,
  \lx_2) \in \gammal(\clx_1) \times \gammal(\clx_2) }).$
\\
\noindent Algorithmically:
\begin{align*}
\top \cjoin \ul &= \ul \cjoin \top = \top\\
\clx \cjoin \ul &= \ul \cjoin \clx = \ul \text{  if } \clx \neq \top\\
\lx_1 \cjoin \lx_2 &= \lx_1 \ljoin \lx_2
\end{align*}
\end{definition}

Both gradual label stamping and gradual join of security types
are obtained by lifting their corresponding static operations:

\begin{definition}[Gradual label meet]
$\\ \clx_1 \cmeet \clx_2 = \alphal(\set{ \lx_1 \lmeet \lx_2 | (\lx_1,
  \lx_2) \in \gammal(\clx_1) \times \gammal(\clx_2) })$.
\end{definition}
\noindent Algorithmically:
\begin{align*}
\bot \cmeet \ul &= \ul \cmeet \bot = \bot\\
\clx \cmeet \ul &= \ul \cmeet \clx = \ul \text{  if } \clx \neq \bot\\
\lx_1 \cmeet \lx_2 &= \lx_1 \lmeet \lx_2
\end{align*}

We now lift subtyping to gradual security types.
\begin{definition}[Consistent subtyping]
$  \cS_1 \csub \cS_2  \text{ if and only if} \\
    \qquad S_1 \sub S_2 \text{ for some } (S_1, S_2) \in
    \gammas(\cS_1) \times \gammas(\cS_2)$
  \end{definition}

\begin{figure}[h]

  \begin{small}
  \begin{displaymath}
    \begin{array}{rcll}
      \multicolumn{4}{c}{
        \clx \in \GLabel,\quad
        \cS \in \GType,\quad
        x \in \oblset{Var},\quad
        b \in \oblset{Bool},\quad
        \oplus \in \oblset{BoolOp}
      }\\
      \multicolumn{4}{c}{
        \ct \in \GTerm,\quad
      
        r \in \oblset{RawValue}\quad      
        v \in \oblset{Value}\quad
        \Gamma \in \oblset{Var} \finto \GType
      } \\[3mm]
      \cS & ::= & \Bool_\clx | \cS ->_\clx \cS & \text{(gradual types)} \\
      b & ::= & \ttt | \ff & \text{(Booleans)}\\
      r & ::= & b | \lambda x:\cS.\ct & \text{(base values)}\\
      v & ::= & r_\lx & \text{(values)}\\
      \ct & ::= & v |  \ct\;\ct | \ct \oplus \ct | \ite{\ct}{\ct}{\ct}
      & \text{(terms)} \\
      \oplus & ::= & \land | \lor | \implies & \text{(operations)} \\
    \end{array}
  \end{displaymath}

\framebox{$\Gamma |- t : S$}
 \begin{mathpar}
    \inference[($\cS$x)]{x:\cS\in\Gamma}{\Gamma |- x : \cS}
    \and
    \inference[($\cS$b)]{}{\Gamma |- b_\clx : \Bool_\clx}
    \and
    \inference[($\cS\lambda$)]{
      \Gamma, x:\cS_1 |- \ct : \cS_2
    }{
      \Gamma |- (\lambda x : \cS_1.\ct)_\clx : \cS_1 ->_\clx \cS_2
    }
    \and
    \inference[($\cS\oplus$)]{
      \Gamma |- \ct_1 : \Bool_{\clx_1} & 
      \Gamma |- \ct_2 : \Bool_{\clx_2} &
    }{
      \Gamma |- \ct_1 \oplus \ct_2 : \Bool_{\clx_1 \cljoin \clx_2}
    }
    \and
    \inference[($\cS$app)]{
      \Gamma |- \ct_1 : \cS_{11} ->_\clx \cS_{12} &
      \Gamma |- \ct_2 : \cS_2 & \cS_2 \csub \cS_{11}
    }{
      \Gamma |- \ct_1\;\ct_2 : \cS_{12} \cljoin \clx
    }
    \and
    \inference[($\cS$if)]{
      \Gamma |- \ct : \Bool_\clx &
      \Gamma |- \ct_1 : \cS_1 & \Gamma |- \ct_2 : \cS_2
    }{
      \Gamma |- \ite{\ct}{\ct_1}{\ct_2} : (\cS_1 \cssubjoin \cS_2) \cljoin \clx
    }
    \and
    \inference[($\cS$::)]{
      \Gamma |- t : \cS_1 &
      \cS_1 \csub \cS_2
    }{
      \Gamma |- t :: \cS_2 : \cS_2
    }
  \end{mathpar}

\framebox{$\cS \csub \cS$}
\begin{mathpar} 
  \inference {\clx \csubl \clx'}{\Bool_\clx \csub \Bool_{\clx'}}
  \and
  \inference
  { \cS'_1 \csub \cS_1 & \cS_2 \csub \cS'_2 & \clx \csubl \clx'}
  { \cS_1 ->_\clx \cS_2 \csub \cS'_1 ->_{\clx'} \cS'_2}
\end{mathpar}

\framebox{$S \cssubjoin S$, $S \cssubmeet S$}
\begin{equation*}
  \begin{block}
    \cssubjoin : \GType \times \GType \rightharpoonup \GType \\[0.2em]
    \Bool_\clx \cssubjoin \Bool_{\clx'} = \Bool_{(\clx \cljoin \clx')} \\
    (\cS_{11} ->_\clx \cS_{12}) \cssubjoin (\cS_{21} ->_{\clx'} \cS_{22}) = 
    (\cS_{11} \cssubmeet \cS_{21}) ->_{(\clx \cljoin \clx')} (\cS_{12} \cssubjoin \cS_{22}) \\
    \cS \cssubjoin \cS \text{ undefined otherwise} \\
    \\
    \cssubmeet : \GType \times \GType \rightharpoonup \GType \\[0.2em]
    \Bool_\clx \cssubmeet \Bool_{\clx'} = \Bool_{(\clx \lmeet \clx')} \\
    (\cS_{11} ->_{\clx} \cS_{12}) \cssubmeet (\cS_{21} ->_{\clx'} \cS_{22}) = 
    (\cS_{11} \cssubjoin \cS_{21}) ->_{(\clx \clmeet \clx')} (\cS_{12} \cssubmeet \cS_{22}) \\
    \cS \cssubmeet \cS \text{ undefined otherwise} \\
  \end{block}
\end{equation*}

\end{small}
 \caption{\lgsec: Syntax and Static Semantics}
  \label{fig:lgsec}
\end{figure}

\subsection{Gradual Security Type System}
The gradual security type system is adapted from Figure~\ref{fig:lsec}
by lifting static types and labels to gradual types and labels, 
lifting partial functions on static types to partial functions on gradual
types, and lifting predicates on types and labels to consistent predicates on
gradual types and labels.

The AGT approach yields a gradual counterpart to an underlying static type
system that satisfies a number of desirable properties.  To state these
properties, the following propositions use $|-_S$ to denote the \lsec typing
relation of Figure~\ref{fig:lsec}.

\begin{restatable}[Conservative Extension]{proposition}{sgequivs}
  \label{prop:sgequivs}
  \mbox{}
  For ${t \in \Term}$, \\
  $|-_S t : S$ if and only if ~$|- t : S$
\end{restatable}
\begin{proof}
  By induction over the typing derivations.  The proof is trivial because
  static types are given singleton meanings via concretization.
\end{proof}

In the following proposition, precision on terms $\ct_1 \gprec \ct_2$ is the
natural lifting of type precision to terms.

\begin{restatable}[Static gradual guarantee]{proposition}{sgradone}
  \label{prop:sgrad-i}
If $|- \ct_1 : \cS_1$ and $\ct_1 \gprec \ct_2$, then $|- \ct_2 : \cS_2$ and $\cS_1 \gprec \cS_2$.
\end{restatable}
\begin{proof}
  Corollary of the corresponding proposition for open terms.  By induction on
  typing derivation of $\Gamma |- \ct_1 : \cS_1$ using the definition of
  $\ct_1 \gprec \ct_2$.
\end{proof}

\subsection{Dynamic Semantics of Gradual Security Typing}

\paragraph{Interiors of consistent subtyping and label ordering.}

The \emph{interior} of a consistent judgment expresses the most precise
deducible information about a consistent judgment.  We define the interior of a
judgment in terms of our abstraction.

\begin{definition}[Interior]
  Let $P$ be a binary predicate on static types.
  Then the \emph{interior} of the judgment
  $\cP(\cT_1,\cT_2)$, notation $\interior{P}(\cT_1,\cT_2)$, is the smallest
  tuple $\evpr{\cT'_1,\cT'_2} \sqsubseteq^2 \evpr{\cT_1,\cT_2}$
  such that for $\evpr{T_1,T_2} \in \Type^2$, if
  $\evpr{T_1,T_2} \in \gamma^2(\cT_1,\cT_2)$ and $P(T_1,T_2)$, then
  $\evpr{T_1,T_2} \in \gamma^2(\cT'_1,\cT'_2)$.  

  It is formalized as follows:
  \begin{gather*}
    \interior{P}(\cT_1,\cT_2) = 
    \alpha^2(\set{\braket{T_1,T_2} \in 
      \gamma^2(\cT_1,\cT_2) | P(T_1,T_2)}).
  \end{gather*}
\end{definition}

We use case-based analysis to calculate the algorithmic rules for the interior
of consistent subtyping on gradual security types:
\begin{mathpar}
  \inference{
    \interior{\subl}(\clx_1, \clx_2) = \braket{\clx'_1, \clx'_2}
  }{
    \Isub(\Bool_{\clx_1}, \Bool_{\clx_2}) =
    \braket{\Bool_{\clx'_1}, \Bool_{\clx'_2}}
  }
  \and
  \inference{
    \Isub(\cS_{21},\cS_{11}) = \braket{\cS'_{21},\cS'_{11}} &
    \Isub(\cS_{12},\cS_{22}) = \braket{\cS'_{12},\cS'_{22}} \\
    \interior{\subl}(\clx_1, \clx_2) = \braket{\clx'_1, \clx'_2}
  }{
    \begin{block}
      \Isub(\cS_{11} ->_{\clx_1} \cS_{12},\cS_{21} ->_{\clx_2} \cS_{22})  = 
      \braket{\cS'_{11} ->_{\clx'_1} \cS'_{12},
        \cS'_{21} ->_{\clx'_2} \cS'_{22} }
    \end{block}
  }
\end{mathpar}

The rules appeal to the algorithmic rules for the interior of consistent label
ordering, calculated similarly:
\begin{mathpar}
  \inference{\lx \neq \top}{\interior{\subl}(\lx, \?) = \braket{\lx, \?}}
  \and
  \inference{}{\interior{\subl}(\top, \?) = \braket{\top, \top}}
  \and
  \inference{\lx \neq \bot}{\interior{\subl}(\?, \lx) = \braket{\?, \lx}}
  \and
  \inference{}{\interior{\subl}(\?, \bot) = \braket{\bot, \bot}}
  \and
  \inference{}{\interior{\subl}(\clx, \clx) = \braket{\clx, \clx}}
\end{mathpar}

\paragraph{Intrinsic terms.} Fig.~\ref{fig:lgsec-intrinsic} presents the
intrinsic terms for \lgsec.  Note that we do not need to introduce term
stamping in this language. Since terms are intrinsically typed and we have
ascriptions, labels can be stamped at the type level.

\begin{figure}[t]
\begin{small}
\begin{mathpar}
  \inference{}{x^{\cS} \in \TermT{\cS}}
  \and
  \inference{}{b_\clx \in \TermT{\Bool_\clx}}
  \and
  \inference{t^{\cS_2} \in \TermT{\cS_2}
  }{
    (\lambda x^{\cS_1}.t^{\cS_2})_\clx \in \TermT{\cS_1 ->_{\clx} \cS_2}
  }
  \and
  \inference{
    t^{\cS_1} \in \TermT{\cS_1} &     
    \ev_1 |- \cS_1 \csub \Bool_{\clx_1} \\ 
    t^{\cS_2} \in \TermT{\cS_2} &     
    \ev_2 |- \cS_2 \csub \Bool_{\clx_2} 
  }{
   \cast{\ev_1}{t^{\cS_1}} \oplus^{\clx_1 \cjoin \clx_2}
   \cast{\ev_2}{t^{\cS_2}} \in
   \TermT{\Bool_{\clx_1 \cjoin \clx_2 }}
  }
  \and
  \inference{
    t^{\cS_1} \in \TermT{\cS_1} & 
    \ev_1 |- \cS_1 \csub \cS_{11} ->_{\clx} \cS_{12}\\ 
   t^{\cS_2} \in \TermT{\cS_2} & 
   \ev_2 |- \cS_2 \csub \cS_{11}
  }{
    \cast{\ev_1}{\itm{1}} \iapp{\cS_{11} ->_{\clx} \cS_{12}}
    \cast{\ev_2}{\itm{2}}
    \in \TermT{\cS_{12} \cljoin \clx}
  }
  \and
  \inference{
   t^{\cS_1} \in \TermT{\cS_1} & 
   \ev_1 |- \cS_1 \csub \Bool_{\clx_1} \\
   t^{\cS_2} \in \TermT{\cS_2} &    
   \ev_2 |- \cS_2 \csub (\cS_2 \cjoin \cS_3) \cjoin \clx_1\\
   t^{\cS_3} \in \TermT{\cS_3} &   
   \ev_3 |- \cS_3 \csub (\cS_2 \cjoin \cS_3) \cjoin \clx_1\\
  }{
    \ite{\cast{\ev_1}{t^{\cS_1}}}{\cast{\ev_2}{t^{\cS_2}}}
      {\cast{\ev_3}{t^{\cS_3}}} 
    \in \TermT{(\cS_2 \cssubjoin \cS_3) \cjoin \clx_1}
  }
  \and
  \inference{
    t^{\cS_1} \in \TermT{\cS_1} &
    \ev_1 |- \cS_1 \csub \cS_2 \\
  }{
    \cast{\ev_1}{t^{\cS_1}} :: \cS_2 \in \TermT{\cS_2}
  } 
\end{mathpar}
\end{small}
\caption{\lgsec: Gradual Intrinsic Terms}
\label{fig:lgsec-intrinsic}
\end{figure}

\paragraph{Reduction.}

Evaluation uses the \emph{consistent transitivity} operator $\trans{\sub}$ to
combine evidences:
\begin{align*}
\braket{\cS_1, \cS_{21}} \trans{<:} \braket{\cS_{22}, \cS_3} & = 
\merge{<:}(\cS_1, \cS_{21} \meet \cS_{22}, \cS_3)
\end{align*}

First we calculate a recursive meet operator for gradual types:
\begin{align*}
\Bool_{\clx} \meet \Bool_{\clx'} &= \Bool_{\clx \meet \clx'}\\
(\cS_{11} ->_{\clx} \cS_{12}) \meet (\cS_{21} ->_{\clx'} \cS_{22}) &=
(\cS_{11} \meet \cS_{21}) ->_{\clx \meet \clx'} (\cS_{12} \meet \cS_{22})\\
\cS \meet \cS' & \text{ undefined otherwise}
\end{align*}
where $\clx \meet \clx' = \alphal(\gammal(\clx) \setint \gammal(\clx'))$, or
algorithmically:
\begin{align*}
  \clx \meet \? &= \? \meet \clx = \clx\\
  \lx \meet \lx &= \lx \\
  \lx \meet \lx' & \text{ undefined otherwise}
\end{align*}

We calculate a recursive definition for $\merge{\sub}$ by case analysis
on the structure of the second argument, 
\begin{mathpar}
  \inference{
    \merge{\subl}(\clx_1,\clx_2, \clx_3) = \braket{\clx_1,\clx_3}
  }{
    \merge{\sub}(\Bool_{\clx_1}, \Bool_{\clx_2}, \Bool_{\clx_3}) =
    \braket{\Bool_{\clx_1},\Bool_{\clx_3}}
  }
  \and
  \inference{
    \merge{\sub}(\cS_{31},\cS_{21},\cS_{11}) = \braket{\cS_{31},\cS_{11}} \\
    \merge{\sub}(\cS_{12},\cS_{22},\cS_{32}) = \braket{\cS_{12},\cS_{32}} \\
    \merge{\subl}(\clx_1,\clx_2, \clx_3) = \braket{\clx_1,\clx_3}}{
    \begin{block}
      \merge{\sub}(\cS_{11} ->_{\clx_1} \cS_{12},
      \cS_{21} ->_{\clx_2} \cS_{22}, \cS_{31} ->_{\clx_3} \cS_{32}) \\[0.5em]=
      \braket{\cS_{11} ->_{\clx_1} \cS_{12},
        \cS_{31} ->_{\clx_3} \cS_{32}}
    \end{block}
  }
\end{mathpar}
with the following definition of $\merge{\subl}$, again calculated by
case analysis on the middle gradual label:
\begin{mathpar}
\inference{\top \csubl \clx_3}{
  \merge{\subl}(\clx_1, \top, \clx_3) = \braket{\clx_1, \top}}
\and
\inference{\clx_1 \csubl \bot}{
  \merge{\subl}(\clx_1, \bot, \clx_3) = \braket{\bot, \clx_3}}
\and
\inference{\clx_1 \csubl \clx_2 & \clx_2 \csubl \clx_3 & \clx_2
  \not\in \set{\bot,\top}}{
  \merge{\subl}(\clx_1,\clx_2, \clx_3) =
  \braket{\clx_1, \clx_3}}
\end{mathpar}

\begin{figure}[t]
\begin{small}
\begin{equation*}
  \begin{array}{rcl}
    \multicolumn{3}{l}{\ev \in \Evidence, \quad
      \evt \in \EvTerm,\quad
      \evv \in \EvValue,\quad
      u \in \SimpleValue,}

    \\
    \multicolumn{3}{l}{
      t \in \TermT{*},\quad
      v \in \Value,\quad
      g \in \EvFrame,\quad
      f \in \TmFrame
    } \\
    \ev & ::= & \pr{\cS,\cS}\\
    \evt & ::= & \cast{\ev}{t} \\
    \evv & ::= & \cast{\ev}{u} \\
    u & ::= & x | b_\clx | (\lambda x.t)_\clx\\[0.3em]
    v & ::= & u | \cast{\ev}{u} :: \cS \\
    g & ::= & [] \oplus^{\clx} \evt | \evv \oplus^{\clx} [] 
    | []\iapp{\cS}\evt | \evv\iapp{\cS}[]
    | []::\cS \\
    & | & \ite{[]}{\evt}{\evt} \\
    f & ::= & g[ \cast{\ev}{[]} ]
  \end{array}
\end{equation*}

\vspace{0.5em}

\noindent\textbf{Notions of Reduction}\\[0.5em]
\boxed{
  \begin{block}
    --> : \TermT{\cS} \times (\TermT{\cS} \union \set{\error}) \\[0.5em]
  \end{block}
}
\begin{equation*}
\begin{block}
  \cast{\pr{\Bool_{\clx'_1},\Bool_{\clx''_1}}}{(b_1)_{\clx_1}} \oplus^{\clx}
  \cast{\pr{\Bool_{\clx'_2},\Bool_{\clx''_2}}}{(b_2)_{\clx_2}} -->\\[0.5em]
\hspace{0.8cm} 
\cast{\pr{\Bool_{(\clx'_1 \cjoin \clx'_2)},
    \Bool_{(\clx''_1 \cjoin \clx''_2)}}}
       {(b_1 \;\llbracket \oplus \rrbracket\; b_2)}_{(\clx_1
         \cjoin \clx_2)} :: \Bool_\clx
  \\[1.2em]
  \evcast{\ev_1}{(\lambda x^{\cS_{11}}.t^{*})_{\clx_1}}
   \iapp{\cS_1 ->_\clx \cS_2}
  \evcast{\ev_2}{u}
  -->
  \\[0.3em]
\hspace{0.5cm}\begin{cases}
  \evcast{\invcod(\ev_1)}{
    (  [((\ev_2 \trans{<:} \invdom(\ev_1)){u} :: 
    \cS_{11})/x^{\cS_1}] t^{*}) :: \cS_{2} \cjoin \clx} \\
  \error \qquad \text{if }(\ev_2 \trans{<:} \invdom(\ev_1)) \text{ not defined}
  \end{cases}
  \\[1.2em]

 \<if>   \cast{\ev_1}{b_{\clx}} 
  \<then> \cast{\ev_2}{t^{\cS_2}}
  \<else> \cast{\ev_3}{t^{\cS_3}}
  --> \\[0.2em]
  \hspace{3cm}
  \begin{cases}
    \cast{\ev_2}{t^{\cS_2}} 
    :: (\cS_2 \cssubjoin \cS_3) \cjoin \clx
    & b = \mathsf{true} \\
    \cast{\ev_3}{t^{\cS_3}} 
    :: (\cS_2 \cssubjoin \cS_3) \cjoin \clx
    & b = \mathsf{false}
  \end{cases}
\end{block}
\end{equation*}
\\[1em]
\boxed{
  \begin{block}
   -->_c : \EvTerm \times  (\EvTerm \union \set{\error})
  \end{block}
}

\begin{equation*}
\begin{block}
\ev_1(\cast{\ev_2}{v :: \cS}) -->_c 
  \begin{cases}
    (\ev_2 \trans{<:} \ev_1){v} \\
    \error \qquad \text{if not defined}
  \end{cases}
\end{block}
\end{equation*}

\vspace{0.5em}

\textbf{Reduction}\\[0.5em]
\boxed{\red : \TermT{\cS} \times (\TermT{\cS} \union \set{\error})}

\begin{mathpar}
  \inference[(R$-->$)]{t^{\cS} --> r & r \in (\TermT{\cS} \union \set{\error})}{
    t^{\cS} \red r} 
  \and
  \inference[(R$g$)]{\evt -->_c \evt'}{
    g[\evt] \red g[r]} 
  \and
  \inference[(R$g$err)]{\evt --> _c\error}{
    g[\evt] \red \error
  }  \and
  \inference[(R$f$)]{t^{\cS}_1 \red t^{\cS}_2}{
    f[t^{\cS}_1] \red f[t^{\cS}_2]
  }
  \and
  \inference[(R$f$err)]{t^{\cS}_1 \red \error}{
    f[t^{\cS}_1] \red \error
  } 
\end{mathpar}
\end{small}
 \caption{\lgsec: Intrinsic Reduction}
  \label{fig:lgsec-dyn}
\end{figure}

The reduction rules are given in Fig.~\ref{fig:lgsec-dyn}.
The evidence inversion functions reflect the contravariance on
arguments and the need to stamp security labels on return types:
\begin{equation*}
\begin{block}
\invdom(\evpr{\cS'_{1} ->_{\clx'} \cS'_{2},\cS'_{1} ->_{\clx''} \cS''_{2}}) =
\evpr{\cS''_1,\cS'_1} \\[0.5em]
\invcod(\evpr{\cS'_{1} ->_{\clx'} \cS'_{2},\cS''_{1} ->_{\clx''} \cS''_{2}}) =
\evpr{\cS'_2 \cjoin \clx',\cS''_2 \cjoin \clx''} 
\end{block}
\end{equation*}

\subsection{Example}

Consider a simple lattice with two confidentiality levels,
$\lbot=\bot$ and $\ltop=\top$, and the following 
\emph{extrinsic} program definitions:
\begin{displaymath}
\begin{array}{rcll}
f & \triangleq & (\lambda x : {\Bool_\lbot}.x)_\lbot &
\text{a public channel}\\[1mm]
g & \triangleq & (\lambda x : {\Bool_\ul}.x)_\lbot &
\begin{tabular}[t]{@{}l@{}}
  an unknown channel \\
  that can be publicly used
\end{tabular} \\
v & \triangleq & \ff_\ltop &\text{a private value}\\
\end{array}
\end{displaymath}
$f\;v$ does not type check, but $f \; (g\; v)$ \emph{does}
typecheck, even though it fails at runtime.
Type checking yields the corresponding \emph{intrinsic} definitions:
\begin{displaymath}
(\evcast{\braket{\cS_f, \cS_f}}{f})
\iapp{\cS_f}
  \left(\evcast{\braket{\lbot,\lbot}}{(\braket{\cS_g,\cS_g}{g})}
  \iapp{\cS_g}
  (\evcast{\braket{\ltop,\ltop}}{v})) \right) \\[1mm]
\end{displaymath}
where the intrinsic subterms are essentially identical to their extrinsic 
counterparts:
\begin{displaymath}
\begin{array}{rcll}
f & \triangleq & 
(\lambda x^{\Bool_\lbot}.x)_\lbot &  \text{a public channel}\\[1mm]
g & \triangleq &
 (\lambda x^{\Bool_\ul}.x)_\lbot & 
\text{\begin{tabular}[t]{@{}l@{}}
    an unknown channel \\
    that can be publicly used
  \end{tabular}
}\\[1mm]
v & \triangleq &
 \ff_\ltop &\text{a private value}\\
\end{array}
\end{displaymath}

For conciseness, we abbreviate $\Bool_\clx$ with $\clx$, use $\cS_f$ and
$\cS_g$ to refer to the types of $f$ and $g$, and elide the application
operators $\iapp{\cS}$. At each step, we use grey boxes to highlight the focus
of reduction/rewriting, and underline the result.

\begin{small}
\begin{displaymath}
\begin{array}{ll}
& (\braket{\cS_f, \cS_f}f) \; \big(\braket{\lbot,\lbot}((\braket{\cS_g,\cS_g}\Gbox{g}) \;  (\braket{\ltop,\ltop}{v})) \big) \\[1mm]
= & (\braket{\cS_f, \cS_f}f) \; \big(\braket{\lbot,\lbot}(\Gbox{(\braket{\cS_g, \cS_g}\underline{(\lambda x^\ul.x)_\lbot^{\cS_g}}) \;
    (\braket{\ltop,\ltop}{v})}) \big)\\[1mm]
--> & (\braket{\cS_f, \cS_f}f) \; \big(\braket{\lbot,\lbot}(\underline{\cast{\ul,\ul}
      ([ \Gbox{(\braket{\ltop,\ltop} \trans{\sub} \braket{\ul,\ul})}
      {v  :: \ul}/x]x) :: \ul}) \big)\\[1mm]
=_{\trans{\sub}} & (\braket{\cS_f, \cS_f}f) \; \big(\braket{\lbot,\lbot}(\cast{\ul,\ul}
      (\Gbox{[ \underline{\braket{\ltop,\ltop}}{v  :: \ul}/x]x}) :: \ul) \big)\\[1mm]
=_{[t/t]t} & (\braket{\cS_f, \cS_f}f) \; \big(\braket{\lbot,\lbot}(\Gbox{\cast{\ul,\ul}
      (\underline{\braket{\ltop,\ltop}{v  :: \ul}})} :: \ul) \big)\\[1mm]
-->_c & (\braket{\cS_f, \cS_f}f) \; \big(\braket{\lbot,\lbot}(\underline{\Gbox{(\braket{\ltop,\ltop} \trans{\sub}
       \braket{\ul,\ul})}{v}} :: \ul) \big)\\[1mm]
=_{\trans{\sub}} & (\braket{\cS_f, \cS_f}f) \; \big(\Gbox{\braket{\lbot,\lbot}(\underline{\braket{\ltop,\ltop}}{v} :: \ul) }\big)\\[1mm]
-->_c & (\braket{\cS_f, \cS_f}f) \; \big(\underline{\error}\big) \\[1mm]

&\text{because }
  \braket{\ltop,\ltop} \trans{\sub} \braket{\lbot,\lbot} = 
  \merge{\sub}(\ltop, \ltop \meet \lbot, \lbot)\\
& \text{which is undefined because }  \ltop \meet \lbot \text{ is undefined.}
\end{array}
\end{displaymath}
\end{small}

\section{Noninterference for \lgsec}
\label{sec:non-interference}

\begin{figure*}[t]
\begin{displaymath}
\begin{array}{rcl}
  v_1 \rel v_2 : \Bool_\clx & \iff & 
  v_i \in \TermT{\Bool_\clx} \land 
  \clx \csubl \lobs \implies \bval{v_1} = \bval{v_2}\\
  v_1 \rel v_2 : \cS_1 ->_{\clx} \cS_2 & \iff & 
  v_i \in \TermT{\cS_1 ->_{\clx} \cS_2} \land \clx \csubl \lobs \implies
  \forall \ev_1, \ev_2, \cS'_1 ->_{\clx'} \cS'_2, \text{ and }\cS''_1
                                                \text{ such that}~\\[0.8em]
& & 
\ev_1 |-  \cS_1 ->_{\clx} \cS_2 \csub \cS'_1 ->_{\clx'} \cS'_2
    \text{ and } \ev_2 |- \cS''_1 \csub \cS'_1 \text{, we have: }\\[0.5em]
& & \forall v'_1 \rel v'_2 : \cS''_1.~
(\evcast{\ev_1}{v_1} \iapp{\cS'_1 ->_{\clx'} \cS'_2} \evcast{\ev_2}{v'_1}) \rel 
(\evcast{\ev_1}{v_2} \iapp{\cS'_1 ->_{\clx'} \cS'_2} \evcast{\ev_2}{v'_2}) : \rcomp{\cS_2 \cljoin \clx}\\
& & \\
t_1 \rel t_2 : \rcomp{\cS} & \iff & t_i \in
                                        \TermT{\cS}~
                                        \wedge~(t_1 \red^{*} v_1
                                        \wedge t_2 \red^{*} v_2 \implies v_1
                                        \rel v_2 : \cS) 
\end{array}
\end{displaymath}
\caption{Gradual security logical relations}
\label{fig:rels}
\end{figure*}

We establish noninterference for the gradual security language using
logical relations, adapting the technique from~\citet{zdancewic}. 

First, in the intrinsic
setting, the type environment related to an open term $t$ is simply the
set of (intrinsically-typed) free variables of the term, $FV(t)$. We
use the metavariable $\Gamma \in \Pow(\Var_{*})$ to denote such
``type-environments-as-sets''. 

Informally, the noninterference theorem states that a program with
low (visible) output and a high (private) input can be run with
different high-security values and, if terminating, will always yield
the same observable value.\footnote{The $\slabel$ function returns the
top-level security label of the given type:
$\slabel(\Bool_\clx) = \clx$ and $\slabel(\cS_1 ->_\clx \cS_2) = \clx$.} 

\begin{restatable}[Noninterference]{theorem}{ifni}
\label{ifni}
if $t \in \TermT{\Bool_\clx}$ with $FV(t) =
\set{x}, x \in \TermT{\cS}$, and 
$v_1, v_2 \in \TermT{\cS}$ with $\slabel(\cS) \ncsubl \clx$, then
\begin{displaymath}
t[v_1/x] \red^{*} v'_1 \wedge t[v_2/x] \red^{*} v'_2 => \bval{v'_1} = \bval{v'_2}
\end{displaymath}
\end{restatable}
\begin{proof}
The result follows by using the method of logical relations
(following~\citet{zdancewic}), as a special case of
Lemma~\ref{lm:relsubst} below.
\end{proof}

Note that we compare equality of bare values at base types, stripping the
checking-related information (labels, evidences and
ascriptions): \ie~$\bval{b_\clx} = b$ and
$\bval{\cast{\ev}{b_\clx::\cS}} = b$. Also, gradual programs can
fail. We establish
{\em termination-insensitive} noninterference, meaning in particular
that any program
may run into an $\error$ without violating noninterference.

\begin{definition}[Gradual security logical relations]
  For an arbitrary element $\lobs$ of the security lattice, the
  $\lobs$-level gradual security relations are type-indexed binary
  relations on closed terms defined inductively as presented in Figure~\ref{fig:rels}.
The notation $v_1 \rel v_2 : \cS$ indicates that $v_1$ is related to
$v_2$ at type $\cS$ when observed at the security level $\lobs$. Similarly, the notation $t_1 \rel t_2 : \rcomp{\cS}$
indicates that $t_1$ and $t_2$ are related computations that produce
related values at type $\cS$ when observed at the security level $\lobs$.
\end{definition}

The logical relations are very similar to those of~\citet{zdancewic},
except for the points discussed above and the fact that we account
for subtyping in the relation between values at a function type
(recall that our type system is syntax-directed).

\begin{definition}[Secure program]
A well-typed program $t$ that produces a $\lobs$-observable output of
type $\cS$ (\ie~$\slabel(\cS) \csubl \lobs$) is secure iff $t \rel t : \rcomp{\cS}$.
\end{definition}

\begin{definition}[Related substitutions]
Two substitutions $\subst_1$ and $\subst_2$ are related, notation
$\Gamma |- \subst_1 \rel \subst_2$, if $\subst_i |= \Gamma$ and \\
$\forall x^\cS \in \Gamma. \subst_1(x^\cS) \rel \subst_2(x^\cS) : \cS$.
\end{definition}
 
\begin{restatable}[Substitution preserves typing]{lemma}{ifsubstpres}
\label{lm:if-subst-pres}
If $\itm{} \in \TermT{\cS}$ and $ \subst |= FV(\itm{})$ then
$\subst(\itm{}) \in \TermT{\cS}$.
\end{restatable}
\begin{proof}
By induction on the derivation of $\itm{} \in \TermT{\cS}$
\end{proof}

\begin{restatable}[Reduction preserves relations]{lemma}{relred}
\label{lm:relred}
Consider $t_1, t_2 \in \TermT{\cS}$. Posing $t_i \red^{*} t'_i$, we have
$t_1 \rel t_2 : \rcomp{\cS}$ if and only if $t'_1 \rel t'_2 : \rcomp{\cS}$
\end{restatable}
\begin{proof}
Direct by definition of $t_1 \rel t_2 : \rcomp{\cS}$ and transitivity
of $\red^{*}$.
\end{proof}

\begin{restatable}[Canonical forms]{lemma}{ifcf}
\label{lm:ifcf}
Consider a value $v \in \TermT{\cS}$. Then either $v = u$, or $v =
\cast{\ev}{u} :: \cS$ with $u \in \TermT{\cS'}$ and $\ev |- \cS' \csub \cS$. Furthermore:
\begin{enumerate}
\item If $\cS = \Bool_\clx$ then either $v=b_\clx$ or
$v=\cast{\ev}{b_{\clx'} :: \Bool_\clx}$ with $b_{\clx'} \in
\TermT{\Bool_{\clx'}}$ and 
$\ev |- \Bool_{\clx'}
\csub \Bool_\clx$.
\item  If $\cS = \cS_1 ->_{\clx} \cS_2$ then either 
$v=(\lambda x^{\cS_1}.\itm{2})_\clx$ with $\itm{2} \in \TermT{\cS_2}$ or
$v=\cast{\ev}{(\lambda x^{\cS'_1}.t^{\cS'_2})_{\clx'} :: \cS_1 ->_{\clx} \cS_2 }$
with $t^{\cS'_2} \in \TermT{\cS'_2}$ and 
$\ev |-  \cS'_1 ->_{\clx'} \cS'_2
\csub  \cS_1 ->_{\clx} \cS_2$.
\end{enumerate}
\end{restatable}
\begin{proof}
By direct inspection of the formation rules of gradual intrinsic terms (Figure~\ref{fig:lgsec-intrinsic}).
\end{proof}

\begin{restatable}[Ascription preserves relation]{lemma}{relasc}
\label{lm:relasc} Suppose $\ev |- \cS' \csub \cS$.
\begin{enumerate}
\item If $v_1 \rel v_2 : \cS'$ then 
$(\cast{\ev}{v_1 :: \cS}) \rel (\cast{\ev}{v_2 :: \cS}) :\rcomp{\cS}$.
\item If $t_1 \rel t_2 : \rcomp{\cS'}$ then $(\cast{\ev}{t_1 :: \cS}) \rel (\cast{\ev}{t_2 :: \cS}) :\rcomp{\cS}$.
\end{enumerate}
\end{restatable}

\begin{proof}
Following~\citet{zdancewic}, the proof proceeds by induction on the
judgment $\ev |- \cS' \csub \cS$. The difference here is that
consistent subtyping is justified by evidence, and that the terms have
to be ascribed to exploit subtyping. In particular, case 1 above
establishes a computation-level relation because each ascribed term $(\cast{\ev}{v_i ::
  \cS})$ may not be a value: each value $v_i$
is either a bare value $u_i$ or a casted value $\cast{\ev_i}{u_i ::
  \cS_i}$, with $\ev_i |- S_i
\csub \cS$. In
the latter case, $(\cast{\ev}{(\cast{\ev_i}{u_i :: \cS_i}) :: \cS})$
either steps to $\error$ (in which case the relation is vacuously
established), or steps to $\cast{\ev'}{u_i :: \cS}$, which is a
value.
\end{proof}

Noninterference follows directly from the following lemma, which
establishes that substitution preserves the logical relations:
\begin{restatable}[Substitution preserves relation]{lemma}{relsubst}
\label{lm:relsubst}
If $\itm{} \in \TermT{\cS}$, $\Gamma = FV(\itm{})$, and $\Gamma |- \subst_1 \rel
\subst_2$, then $\subst_1(\itm{}) \rel \subst_2(\itm{}) : \rcomp{\cS}$.
\end{restatable}
\begin{proof}
By induction on the derivation that $t \in \TermT{\cS}$. Considering
the last step used in the derivation:

\begin{case}[$x$] $\itm{} = x^\cS$ so $\Gamma = \{ x^\cS \}$. $\Gamma |- \subst_1 \rel
\subst_2$ implies by definition that $\subst_1(x^\cS) \rel
\subst_2(x^\cS) : \cS$.
\end{case}

\begin{center}--------\end{center}
\begin{case}[$b$] $\itm{} = b_\clx$. By definition of substitution,
  $\subst_1(b_\clx) = \subst_2(b_\clx) = b_\clx$. By definition,
  $b_\clx \rel b_\clx : \Bool_\clx$ as required.
\end{case}

\begin{center}--------\end{center}
\begin{case}[$\lambda$]
$\itm{} = (\lambda x^{\cS_1}.t^{\cS_2})_\clx$. Then $\cS = \cS_1 ->_\clx \cS_2$.

\noindent By definition of substitution, assuming $x^{\cS_1} \not\in
\dom(\subst_i)$,  and Lemma~\ref{lm:if-subst-pres}:
\begin{displaymath}
\subst_i(\itm{}) = (\lambda x^{\cS_1}.\subst_i(t^{\cS_2}))_\clx \in
\TermT{\cS}
\end{displaymath}
If $\clx \ncsubl \lobs$ the result holds trivially because all
function values are related in such a cases. Assume
$\clx \csubl \lobs$, and assume two values $v_1$ and $v_2$ such that
$v_1 \rel v_2 : \cS'_1$, with $\ev_2 |- \cS'_1 \csub \cS_1$.  (We omit
the $\iapp{\cS}$ operator in applications below since we simply pick
$\cS = \cS_1 ->_\clx \cS_2$.)\\
We need to show that:
\begin{displaymath}
\begin{array}{rcl}
 & \cast{\ev_1}{(\lambda x^{\cS_1}.\subst_1(t^{\cS_2}))_\clx}  
\; \cast{\ev_2}{v_1} & \\
\rel &
\cast{\ev_1}{(\lambda
x^{\cS_1}.\subst_2(t^{\cS_2}))_\clx}
\; \cast{\ev_2}{v_2} &
: \rcomp{\cS_2 \cljoin \clx}
\end{array}
\end{displaymath}
Each $v_i$ is either a bare value $u_i$ or a casted value
$\cast{\ev_{2i}}{u_i} :: \cS'_1$. In the latter case, the application
expression combines evidence, which may fail with $\error$. If it
succeeds, we call the combined evidence $\ev'_{2i}$. The application rule then
applies: it may fail with $\error$ if the evidence $\ev'_{2i}$ cannot be
combined with the evidence for the function parameter. In all of the failure
cases, the relation vacuously holds. We therefore consider the only interesting
case, where the applications succeed. We have:
\begin{displaymath}
\begin{array}{rl}
 & \cast{\ev_1}{(\lambda x^{\cS_1}.\subst_i(t^{\cS_2}))_\clx} \;
\cast{\ev'_{2i}}{u_i}\\ 
\red & \cast{\ev_r}{([\cast{\ev_a}{u_i ::
                            \cS_1}/x^{\cS_1}] \subst_i(\itm{2}))}
:: \cS_2 \cljoin \clx
\end{array}
\end{displaymath}
where $\ev_r$ and $\ev_{ai}$ are the new evidences for the return
value and argument, respectively.
We then extend the substitutions to map $x^{\cS_1}$ to the casted
arguments:
\begin{displaymath}
\subst'_i = \subst_i \{ x^{\cS_1} \mapsto \cast{\ev_a}{u_i :: \cS_1} \}
\end{displaymath}
By Lemma~\ref{lm:relasc}.1, $(\cast{\ev_{a1}}{u_1 :: \cS_1}) \rel 
(\cast{\ev_{a2}}{u_2 :: \cS_1}) : \cS_1$\\So $\Gamma, x^{\cS_1} |- \subst'_1 \rel \subst'_2$. By induction hypothesis:
\begin{displaymath}
\subst'_1(\itm{2}) \rel \subst'_2(\itm{2}) : \rcomp{\cS_2}
\end{displaymath}
By the definition of substitution, this is exactly:
\begin{displaymath}
[ \cast{\ev_a}{u_1 :: \cS_1}/x^{\cS_1}] \subst_1(\itm{2}) \rel 
[ \cast{\ev_a}{u_2 :: \cS_1}/x^{\cS_1}] \subst_2(\itm{2}) : \rcomp{\cS_2}
\end{displaymath}
Finally, since $\ev_r |- \cS_2 \csub \cS_2 \cljoin \clx$, by
Lemma~\ref{lm:relasc}.2:
\begin{displaymath}
\begin{array}{rcl}
 & \cast{\ev_r}{[ \cast{\ev_a}{u_1 :: \cS_1}/x^{\cS_1}] \subst_1(\itm{2})
  :: \cS_2 \cljoin \clx } & \\
\rel & \cast{\ev_r}{[ \cast{\ev_a}{u_2 :: \cS_1}/x^{\cS_1}] \subst_2(\itm{2})
:: \cS_2 \cljoin \clx} 
& :  \rcomp{\cS_2 \cljoin \clx}
\end{array}
\end{displaymath}
By backward preservation of the relations (Lemma~\ref{lm:relred}), this implies that:
\begin{displaymath}
\begin{array}{rcl}
& \cast{\ev_1}{(\lambda x^{\cS_1}.\subst_1(t^{\cS_2}))_\clx}  
\; \cast{\ev_2}{v_1} & \\
\rel & 
\cast{\ev_1}{(\lambda
x^{\cS_1}.\subst_2(t^{\cS_2}))_\clx}
\; \cast{\ev_2}{v_2}
& : \rcomp{\cS_2 \cljoin \clx}
\end{array}
\end{displaymath}

\end{case}

\begin{center}--------\end{center}
\begin{case}[$\oplus$]
$\itm{} = \cast{\ev_1}{\itm{1}} \;\oplus^\clx\; \cast{\ev_2}{\itm{2}}$\\[0.5em]

\noindent By definition of substitution and Lemma~\ref{lm:if-subst-pres}:
\begin{displaymath}
\subst_i(\itm{}) = 
\cast{\ev_1}{\subst_i(\itm{1})} \;\oplus^\clx\;
\cast{\ev_2}{\subst_i(\itm{2})} \in \TermT{\cS}
\end{displaymath}
By induction hypotheses:
\begin{displaymath}
\subst_1(\itm{1}) \rel \subst_2(\itm{1}) : \rcomp{\cS_1} \text{ and }
\subst_1(\itm{2}) \rel \subst_2(\itm{2}) : \rcomp{\cS_2}
\end{displaymath}
By definition of related computations:
\begin{displaymath}
\begin{array}{c}
\subst_1(\itm{1}) \red^{*} v_{11} \wedge \subst_2(\itm{1}) \red^{*} v_{21}
=> v_{11} \rel v_{21} : \cS_1\\
\subst_1(\itm{2}) \red^{*} v_{12} \wedge \subst_2(\itm{2}) \red^{*} v_{22}
=> v_{12} \rel v_{22} : \cS_2
\end{array}
\end{displaymath}                 
By Lemma~\ref{lm:ifcf}, each $v_{ij}$ is either a boolean
$(b_{ij})_{\clx_{ij}}$ or a
casted boolean $\cast{\ev_{ij}}{(b_{ij})_{\clx'_{ij}}} :: \cS_j$. 
In case a value $v_{ij}$ is a casted value, then the whole term
$\subst_i(\itm{})$ can take a step by (R$g$), combining $\ev_i$ with
$\ev_{ij}$. Such a step either fails, or
succeeds with a new combined evidence. Therefore, either:
\begin{displaymath}
\subst_i(\itm{}) \red^{*} \error
\end{displaymath}
in which case we do not care since
we only consider termination-insensitive noninterference, or:
\begin{displaymath}
\begin{array}{rll}
\subst_i(\itm{}) & \red^{*}&
\cast{\ev'_1}{(b_{i1})_{\clx'_{i1}}} \;\oplus^\clx\;
\cast{\ev'_2}{(b_{i2})_{\clx'_{i2}}}\\[0.3em]
& \red& \cast{\ev'}{(b_i)_{\clx'_i}} :: \Bool_\clx
\end{array}
\end{displaymath}
with $b_i = b_{i1} \llbracket\oplus\rrbracket b_{i2}$ and $\clx'_i =
\clx'_{i1}\cljoin\clx'_{i2}$. It remains to show that:
\begin{displaymath}
(b_1)_{\clx'_1} 
\rel
(b_2)_{\clx'_2} 
: \Bool_\clx
\end{displaymath}
If $\clx \ncsubl \lobs$, then the result holds trivially because
all boolean values are related. If $\clx \csubl \lobs$, then also
$\clx'_i \csubl \lobs$, which means by definition of $\rel$ on boolean
values, that $b_{11} = b_{21}$ and $b_{12} = b_{22}$, so $b_1 = b_2$.

\end{case}

\begin{center}--------\end{center}
\begin{case}[app]

$\itm{} = \cast{\ev_1}{\itm{1}} \iapp{\cS_{11} ->_\clx \cS_{12}}  \cast{\ev_2}{\itm{2}}$ \\
with $\ev_1 |- \cS_1 \csub S_{11} ->_{\clx} S_{12}$, $\ev_2 |- \cS_2
\csub \cS_{11}$,  
and $\cS = \cS_{12} \cljoin \clx$.\\
We omit the $\iapp{\cS_{11} ->_\clx \cS_{12}}$ operator in applications below.\\[0.3em]

\noindent By definition of substitution and Lemma~\ref{lm:if-subst-pres}:
\begin{displaymath}
\subst_i(\itm{}) = 
\cast{\ev_1}{\subst_i(\itm{1})} \;
\cast{\ev_2}{\subst_i(\itm{2})} \in \TermT{\cS}
\end{displaymath}
By induction hypothesis:
\begin{displaymath}
\subst_1(\itm{1}) \rel \subst_2(\itm{1}) : \rcomp{\cS_1} \text{ and }
\subst_1(\itm{2}) \rel \subst_2(\itm{2}) : \rcomp{\cS_2}
\end{displaymath}
By definition of related computations:
\begin{displaymath}
\begin{array}{c}
\subst_1(\itm{1}) \red^{*} v_{11} \wedge \subst_2(\itm{1}) \red^{*} v_{21}
=> v_{11} \rel v_{21} : \cS_1 \\
\subst_1(\itm{2}) \red^{*} v_{12} \wedge \subst_2(\itm{2}) \red^{*} v_{22}
=> v_{12} \rel v_{22} : \cS_2
\end{array}
\end{displaymath}
By definition of $\rel$ at values of function type, using $\ev_1$ and
$\ev_2$ to justify the subtyping relations, we have:
\begin{displaymath}
 (\cast{\ev_1}{v_{11}} \; \cast{\ev_2}{v_{12}}) \rel 
  (\cast{\ev_1}{v_{21}} \; \cast{\ev_2}{v_{22}})  : \rcomp{\cS_{12}
    \cljoin \clx}
\end{displaymath}
\end{case}

\begin{center}--------\end{center}
\begin{case}[\textsf{if}]
$\itm{} =
\ite{\cast{\ev_1}{\itm{1}}}{\cast{\ev_2}{\itm{2}}}{\cast{\ev_3}{\itm{3}}}$,
with 
$\itm{i} \in \TermT{\cS_i}$, 
$\ev_1 |- \cS_1 \csub \Bool_{\clx_1} $ and
$\cS = (\cS_2 \cssubjoin \cS_3) \cjoin \clx_1 $\\
By definition of substitution:
\begin{displaymath}
\subst_i(\itm{}) = 
\ite{\cast{\ev_1}{\subst_i(\itm{1})}}{\cast{\ev_2}{\subst_i(\itm{2})}}
{\cast{\ev_3}{\subst_i(\itm{3})}}
\end{displaymath}
If $\clx \ncsubl \lobs$, then 
$\subst_1(\itm{})  \rel \subst_2(\itm{}) : \rcomp{\cS}$
holds
trivially because the $\rel$ relations relate all such well-typed
terms.
Let us assume $\clx \csubl \lobs$. By the induction hypothesis we have
that:
\begin{displaymath}
\subst_1(\itm{1}) \rel \subst_2(\itm{1}) : \rcomp{\cS_1}
\end{displaymath}
Assuming $\subst_i(\itm{1}) \red^{*} v_{i1}$, by the definition of
$\rel$ we have:
\begin{displaymath}
v_{11} \rel v_{21} : \cS_1
\end{displaymath}
By Lemma~\ref{lm:ifcf}, each $v_{i1}$ is either a boolean
$(b_{i1})_{\clx_{i1}}$ or a
casted boolean $\cast{\ev_{i1}}{(b_{i1})_{\clx'_{i1}}} :: \cS_1$. 
In either case, $\cS_1 \csub \Bool_{\clx_1}$ implies $\cS_1 =
\Bool_{\clx'_1}$, so by definition of $\rel$ on boolean values, 
$b_{11} = b_{21}$.\\
In case a value $v_{i1}$ is a casted value, then the whole term
$\subst_i(\itm{})$ can take a step by (R$g$), combining $\ev_i$ with
$\ev_{i1}$. Such a step either fails, or
succeeds with a new combined evidence. Therefore, either:
\begin{displaymath}
\subst_i(\itm{}) \red^{*} \error
\end{displaymath}
in which case we do not care since
we only consider termination-insensitive noninterference, or:
\begin{displaymath}
\begin{array}{rll}
\subst_i(\itm{}) & \red^{*}&
\ite{\cast{\ev'_1}{(b_{i1})_{\clx'_1}}}
{\cast{\ev_2}{\subst_i(\itm{2})}}
{\cast{\ev_3}{\subst_i(\itm{3})}} 
\end{array}
\end{displaymath}
Because $b_{11} = b_{21}$, both $\subst_1(\itm{})$ and
$\subst_2(\itm{})$ step into the same branch of the conditional. Let us
assume the condition is true (the other case is similar). Then:
\begin{displaymath}
\subst_i(\itm{}) \red \cast{\ev_2}{\subst_i(\itm{2}) :: \cS}
\end{displaymath}
By induction hypothesis:
\begin{displaymath}
\subst_1(\itm{2}) \rel \subst_2(\itm{2}) : \rcomp{\cS_2}
\end{displaymath}
Assume $\subst_i(\itm{2}) \red^{*} v_{i2}$, then $v_{12} \rel v_{22} :
\cS_2$. Since $\ev_2 |- \cS_2 \csub \cS$, by Lemma~\ref{lm:relasc} we have:
\begin{displaymath}
(\cast{\ev_2}{\subst_1(\itm{2}) :: \cS}) \rel
(\cast{\ev_2}{\subst_2(\itm{2}) :: \cS}) :: \rcomp{\cS}
\end{displaymath}
\end{case}

\begin{center}--------\end{center}
\begin{case}[::]
$\itm{} = \cast{\ev}\itm{1}::\cS$, with $\itm{1} \in \TermT{\cS_1}$
and $\ev_1 |- \cS_1 \csub \cS$. 

\noindent By definition of substitution:
\begin{displaymath}
\subst_i(\itm{}) = \cast{\ev_1}\subst_i(\itm{1})::\cS
\end{displaymath}
By induction hypothesis:
\begin{displaymath}
\subst_1(\itm{1}) \rel \subst_2(\itm{1}) : \rcomp{\cS_1}
\end{displaymath}
The result follows directly by Lemma~\ref{lm:relasc}.

\end{case}
\end{proof}

\section{Related Work and Conclusion}

The design of a gradual security-typed language is a novel
contribution. Despite the fact that both \citet{disney11flow} and
\citet{fennellThiemann:csf2013} have proposed languages for security
typing dubbed gradual, they do not propose gradual source languages.
Rather, the language designs require explicit security casts---which
can also be encoded with a label test expression in
Jif~\cite{zhengMyers:ijis2007}.
Furthermore, both designs treat an unlabeled type as having the top label, then
allowing explicit casts downward in the security lattice.  This design is
analogous to the internal language of the quasi-static typing approach.  In
that approach, explicit casts work well, but the external language there
accepts too many programs.  That difficulty was the original motivation for
consistency in gradual typing~\cite{GradualTyping}.

\citet{thiemannFennell:esop2014} develop a generic approach to gradualize
annotated type systems. This is similar to security typing (labels are one kind
of annotation), except that they only consider annotation on base types, and
the language only includes explicit casts, like the gradual security work
discussed above.  They track blame and provide a translation that removes
unnecessary casts.

\paragraph{Acknowledgments} We thank Mat{\'i}as Toro for feedback and contributing
the proofs of Propositions 1 and 2 in Appendix~\ref{sec:auxiliary-proofs}.

\bibliographystyle{abbrvnat}
\bibliography{references}{}

\appendix

\section{Auxiliary Proofs}
\label{sec:auxiliary-proofs}

\syntypein*
\begin{proof}
  By induction on $\Gamma |- t : S$.  

  As most of the type rules are identical, most of the cases are
  straightforward. The exceptions to this are the (Sif) and (Sapp) rules.
  \begin{case}[Sif]
    Then \\
    \begin{equation*}
      \D = 
      \infer{
        \Gamma |- \ite{t_0}{t_1}{t_2} : (S_1 \subjoin S_2) \ljoin \lx
        }{
          \deduce{\Gamma |- t_0 : \Bool_{\lx}}{\D_0} &
          \deduce{\Gamma |- t_1 : S_1}{\D_1} &
          \deduce{\Gamma |- t_2 : S_2}{\D_2}
        }
      \end{equation*}
      By Lemma~\ref{subtypesubjoin}, $S_1 <: (S_1 \subjoin S_2)$ and
      $S_2 <: (S_1 \subjoin S_2)$, and by Lemma~\ref{subtypeljoin}
      $(S_1 \subjoin S_2) <: (S_1 \subjoin S_2) \ljoin \lx$, therefore
      $S_1 <: (S_1 \subjoin S_2) \ljoin \lx$ and
      $S_2 <: (S_1 \subjoin S_2) \ljoin \lx$.\\
      Combining these with the induction hypotheses, we get
      \begin{equation*}
        \mathcal{E} = 
        \infer[]{
          \Gamma |-^{Z} \ite{t_0}{t_1}{t_2} : (S_1 \subjoin S_2) \ljoin \lx
        }{
          \deduce{\Gamma |-^{Z} t_0 : \Bool_{\lx}}{\E_0} &
          \overline{
            \infer[]{ %
              \Gamma |-^{Z} t_i : (S_i \subjoin S_2) \ljoin \lx
            }{
              \deduce{\Gamma |-^{Z} t_i : S_i}{\E_i} & 
              S_i <: (S_1 \subjoin S_2) \ljoin  \lx
            }}
        }
      \end{equation*}

  \end{case}

  \begin{case}[Sapp]
    Then $t = t_1\;t_2$ and $\Gamma |- t_1\;t_2: S_{12} \ljoin \lx$ for some $S_{12}$ and $\lx$ such that $\Gamma |- t_1 : S_{11} ->_{\lx} S_{12}$, $\Gamma |- t_2 : S_2$ and $S_2 <: S_{11}$. Using induction hypothesis on $t_2$ we know that $\Gamma |-^{Z} t_2 : S_2$. As $S_2 <: S_{11}$. Then by (\lsec-SUB) $\Gamma |-^{Z} t_2 : S_{11}$.
    Using induction hypothesis on $t_1$, $\Gamma |-^{Z} t_1 : S_{11} ->_{\lx} S_{12}$, then by (\lsec-APP) we conclude that $\Gamma |-^{Z} t_1\;t_2: S_{12} \ljoin \lx$
  \end{case}

\end{proof}

\begin{lemma}
  \label{subtypesubjoin}
  Let $S_1, S_2 \in \Type$.  Then
  \enumerate
    \item If $(S_1 \subjoin S_2)$ is defined then $S_1 <: (S_1 \subjoin S_2)$.
    \item If $(S_1 \submeet S_2)$ is defined then $(S_1 \submeet S_2) <: S_1$.
  \endenumerate
\end{lemma}
\begin{proof}
  We start by proving (1) assuming that $(S_1 \ljoin S_2)$ is defined. We proceed by case analysis on $S_1$.
  \begin{case}[$\Bool_\lx$]
    If $S_1 = \Bool_{\lx_1}$ then as $(S_1 \subjoin S_2)$ is defined then $S_2$ must have the form $\Bool_{\lx_2}$ for some $\lx_2$. Therefore $(S_1 \subjoin S_2) = \Bool_{(\lx_1 \ljoin \lx_2)}$. But by definition of $\subl$, $\lx_1 \subl (\lx_1 \ljoin \lx_2)$ and therefore we use ($<:_{\Bool}$) to conclude that $\Bool_{\lx_1} <: \Bool_{(\lx_1 \ljoin \lx_2)}$, i.e. $S_1 <: (S_1 \subjoin S_2)$.
  \end{case}

  \begin{case}[$S ->_\lx S$]
    If $S_1 = S_{11} ->_{\lx_1} S_{12}$ then as $(S_1 \subjoin S_2)$ is defined then $S_2$ must have the form $S_{21} ->_{\lx_2} S_{22}$ for some $S_{21}, S_{22}$ and $\lx_2$.\\
    We also know that $(S_1 \subjoin S_2) = (S_{11} \submeet S_{21}) ->_{(\lx_1 \ljoin \lx_2)} (S_{12} \submeet S_{22})$. By definition of $\subl$, $\lx_1 \subl (\lx_1 \ljoin \lx_2)$. 
    Also, as $(S_1 \subjoin S_2)$ is defined then $(S_{11} \submeet S_{21})$ is defined. 
    Using the induction hypothesis of (2) on $S_{11}$, $(S_{11} \submeet S_{21}) <: S_{11}$. Also, using the induction hypothesis of (1) on $S_{12}$ we also know that $S_{12} <: (S_{12} \submeet S_{22})$.
    Then by ($<:_{->}$) we can conclude that $S_{11} ->_{\lx_1} S_{12} <: (S_{11} \submeet S_{21}) ->_{(\lx_1 \ljoin \lx_2)} (S_{12} \submeet S_{22})$, i.e. $S_1 <: (S_1 \subjoin S_2)$.
  \end{case}

  The proof of (2) is  similar to (1) but using the argument that\\ ${(\lx_1 \lmeet \lx_2) \subl \lx_1}$.
\end{proof}

\begin{lemma}
  \label{subtypeljoin}
  Let $S \in \Type$ and $\lx \in \Label$. Then $S <: S \ljoin \lx$.
\end{lemma}
\begin{proof}
  Straigthforward case analysis on type $S$ using the fact that $\lx \subl (\lx' \ljoin \lx)$ for any $\lx'$.
\end{proof}

\begin{lemma}
  \label{lemma:subtypetwoljoin}
  Let $S_1, S_2 \in \Type$ such that $S_1 <: S_2$, and let $\lx_1, \lx_2 \in \Label$ such that $\lx_1 \subl \lx_2$. Then $S_1 \ljoin \lx_1 <: S_2 \ljoin \lx_2$.
\end{lemma}
\begin{proof}
  Straightforward case analysis on type $S$ using the definition of \emph{label stamping} on types.
\end{proof}

\insyntype*
\begin{proof}
  By induction on derivations of $\Gamma |-^{Z} t : S$.

  We proceed by case analysis on $t$ (modulo (\lsec-SUB)). As most of the type
  rules are identical, most of the cases are straightforward. The exception to
  this is case (\lsec-COND) and (\lsec-APP):
  \begin{case}[\lsec-COND]
    Then $t = \ite{t'}{t_1}{t_2}$ and \\
    ${\Gamma |-^{Z} \ite{t'}{t_1}{t_2} : S_z \ljoin \lx}$ for some $S_z$, and $l$ such that $\Gamma |-^{Z} t' : \Bool_\lx$, $\Gamma |-^{Z} t_1 : S \ljoin \lx$ and $\Gamma |-^{Z} t_2 : S \ljoin \lx$.\\
    Using induction hypothesis on the premises we also know that $\Gamma |- t' : \Bool_{\lx'}$ for some $\lx' \subl \lx$, $\Gamma |- t_1 : S_1'$ for some $S_1' <: S \ljoin \lx$, and $\Gamma |- t_2 : S_2'$ for some $S_2' <: S \ljoin \lx$.\\
    By (Sif), $\Gamma |- \ite{t'}{t_1}{t_2} : (S_1' \ljoin S_2') \ljoin \lx'$.
    Then by Lemma~\ref{subtypesubjoin} we know that $(S_1' \ljoin S_2') <: S_z$, and by Lemma~\ref{lemma:subtypetwoljoin} if we choose $S' = (S_1' \ljoin S_2') \ljoin \lx'$, we conclude that $(S_1' \ljoin S_2') \ljoin \lx' <: S_z \ljoin \lx$, i.e. $S' <: S$.
  \end{case}

  \begin{case}[\lsec-APP]
    Then $t = t_1\;t_2$ and $\Gamma |-^{Z} t_1\;t_2 : S_{12} \ljoin \lx$ for some $S_{12}$ and $\lx$ such that $\Gamma |-^{Z} t_1 : S_{11} ->_{\lx} S_{12}$ and $\Gamma |-^{Z} t_2 : S_{11}$.\\
    Using induciton hypothesis on the premises we also know that $\Gamma |- t_1 : S_{11}' ->_{\lx'} S_{12}'$ for some $S_{11}', S_{12}'$ and $\lx'$ such that $S_{11} <: S_{11}', S_{12}' <: S_{12}$ and $\lx' \subl \lx$, and that $\Gamma |- t_2 : S_{12}''$ such that $S_{11}'' <: S_{11}$. By transitivity on subtyping then $S_{11}'' <: S_{11}'$. Then by (Sapp) $\Gamma |- t_1\;t_2 : S_{12}' \ljoin \lx'$ and by Lemma~\ref{lemma:subtypetwoljoin} if we choose $S' = S_{12}' \ljoin \lx'$, we conclude that $S_{12}' \ljoin \lx' <: S_{12} \ljoin \lx$, i.e. $S' <: S$.
  \end{case}

\end{proof}

\begin{lemma}
 \label{lemma:subtypesubjoin}
 Let $S_1, S_2$ and $S_3 \in \Type$. 
 \enumerate
  \item If $(S_1 \subjoin S_2)$ is defined, $S_1 <: S_3$ and $S_1 <: S_3$ then $(S_1 \subjoin S_2) <: S_3$.
  \item If $(S_1 \submeet S_2)$ is defined, $S_3 <: S_1$ and $S_3 <: S_2$ then $S_3 <: (S_1 \submeet S_2)$.
 \endenumerate
\end{lemma}
\begin{proof}
  We start by proving (1) by case analysis on type $(S_1 \subjoin S_2)$.
  \begin{case}[$\Bool_{\lx}$]
    As $(S_1 \subjoin S_2)$ is defined then $S_1 = \Bool_{\lx_1}$ and $S_2 = \Bool_{\lx_2}$ for some $\lx_1$ and $\lx_2$. Also as $S_1 <: S_3$, then $S_3 = \Bool_{\lx_3}$ for some $\lx_3$. 
    By ($<:_{\Bool}$), $\lx_1 \subl \lx_3$ and $\lx_2 \subl \lx_3$ then by definition of $\ljoin$, $(\lx_1 \ljoin \lx_2) \subl \lx_3$. Then $(S_1 \subjoin S_2) = \Bool_{(\lx_1 \ljoin \lx_2)} <: \Bool_{\lx_3}$, i.e.  $(S_1 \subjoin S_2) <: S_3$.
  \end{case}
  \begin{case}[$S ->_{\lx} S$]
    As $(S_1 \subjoin S_2)$ is defined then $S_1 = S_{11} ->_{\lx_1} S_{12}$ and $S_2 = S_{21} ->_{\lx_1} S_{22}$ for some $S_{11}, S_{12}, S_{21}, S_{22}, \lx_1$ and $\lx_2$. Also as $S_1 <: S_3$, then $S_3 = S_{31} ->_{\lx_3} S_{32}$ for some $S_{31}, S_{32}$ and $\lx_3$.
    By ($<:_{->}$), $S_{31} <: S_{11}, S_{12} <: S_{32}, S_{31} <: S_{21}, S_{22} <: S_{32}, \lx_1 \subl \lx_3$ and $\lx_2 \subl \lx_3$.\\
    Then $(S_1 \subjoin S_2) = (S_{11} \submeet S_{21}) ->_{(\lx_1 \ljoin \lx_2)} (S_{12} \submeet S_{22})$. By using induction hypothesis (2) then $S_{31} <: (S_{11} \submeet S_{21})$ and by using induction hypothesis (1) then $(S_{12} \submeet S_{22}) <: S_{32}$. Also by definition of $\ljoin$, $(\lx_1 \ljoin \lx_2) \subl \lx_3$. Finally by ($<:_{->}$) we conclude that $(S_{11} \submeet S_{21}) ->_{(\lx_1 \ljoin \lx_2)} (S_{12} \submeet S_{22}) <: S_{31} ->_{\lx_3} S_{32}$, i.e.  $(S_1 \subjoin S_2) <: S_3$.
  \end{case}

  The proof of (2) is  similar to (1) but using the argument that if $\lx_3 \subl \lx_1$ and $\lx_3 \subl \lx_2$ then   ${\lx_3 \subl (\lx_1 \lmeet \lx_2)}$.
\end{proof}

\end{document}

%% file: definitions.tex
\usepackage{anyfontsize} % kill warnings about font sizes 
\usepackage{amsmath}     % Mathematics FTW!
\usepackage{amsthm}      % Warning, I've only proven it correct :).
\usepackage{amssymb}     % needed for \mathbb and mathcal
\usepackage{amsbsy}      % for \boldsymbol
\usepackage{semantic}    % for PLT
\usepackage{stmaryrd}   % provides some PLy symbols, like Scott Brackets
\usepackage{mathpartir}  % for \mathpar
\usepackage{color}
\usepackage{braket}      % for \set and \braket
\usepackage{mathtools}   % for \bigtimes
\usepackage{thmtools}    % to duplicate theorems in the appendix
\usepackage{marvosym} % for deniarius (gradual label)
\usepackage{wasysym}     % for \smiley
\usepackage{balance}     % for the references
\usepackage{xspace}
\usepackage{mathrsfs} % for mathscr font
\usepackage{wasysym}
\usepackage{proof}
\usepackage{yhmath}
\usepackage{rotating}

\newenvironment{block}[1][t]
  {\begin{array}[#1]{@{}l@{}}}
  {\end{array}}

\definecolor{lightgray}{gray}{0.90}
\newcommand{\Gbox}[1]{\colorbox{lightgray}{$#1$}}

\declaretheorem{theorem}
\declaretheorem[sibling=theorem]{lemma}

\declaretheorem{definition}
\declaretheorem[style=remark,numbered=no]{case}

\mathlig{[]}{\square}
\mathlig{|-s}{\vdash^{\!\forall}}
\mathlig{|-}{\vdash}
\mathlig{|}{\mid} % For BNFs
\mathlig{@>}{\leadsto} % for type-directed translations

\reservestyle{\oblang}{\mathsf}
\oblang{if[if\;],then[\;then\;],else[\;else\;],let[let\;],in[\;in\;]}

\newcommand{\Pow}{\ensuremath{\mathcal{P}}} % The powerset
\newcommand{\Bool}{\mathsf{Bool}}
\newcommand{\oblset}[1]{\textsc{#1}}
\newcommand{\Label}{\oblset{Label}}
\newcommand{\Type}{\oblset{Type}}
\newcommand{\Var}{\oblset{Var}}
\newcommand{\Value}{\oblset{Value}}
\newcommand{\Term}{\oblset{Term}}

\newcommand{\GTerm}{\oblset{GTerm}}

\newcommand{\GType}{\oblset{GType}}

\newcommand{\GLabel}{\oblset{GLabel}}
\newcommand{\setof}[1]{\set{\overline{#1}}}

\newcommand{\fff}{\textsf{false}}
\newcommand{\ff}{\fff} % temporary - backward compat
\newcommand{\ttt}{\textsf{true}}
\newcommand{\ite}[3]{\<if> #1 \<then> #2 \<else> #3}
\newcommand{\lsec}{$\lambda_\text{SEC}$\xspace}
\newcommand{\lgsec}{$\lambda_{\consistent{\text{SEC}}}$\xspace}

\newcommand{\dom}{\mathit{dom}}
\newcommand{\cod}{\mathit{cod}}

\newcommand{\valid}{\mathit{valid}}

\newcommand{\gprec}{\sqsubseteq} % greater precision
\newcommand{\sub}{<:}

\newcommand{\csub}{\lesssim}

\newcommand{\meet}{\sqcap}
\newcommand{\union}{\cup}
\newcommand{\setint}{\cap}
\newcommand{\finto}{\stackrel{\text{fin}}{\rightharpoonup}}% Finite partial map
\newcommand{\?}{\textsf{\upshape ?}} %unknown type
\newcommand{\consistent}[1]{\widetilde{#1}}
\newcommand{\collecting}[1]{\wideparen{#1}}
\newcommand{\cT}{\consistent{T}} %gradual type
\newcommand{\clT}{\collecting{T}}% collecting type
\newcommand{\ct}{\consistent{t}} %gradual term
\newcommand{\clcod}{\collecting{\cod}}
\newcommand{\cldom}{\collecting{\dom}}
\newcommand{\cP}{\consistent{P}}

\newcommand{\subjoin}{\mathbin{\begin{turn}{90}$\hspace{-0.2em}<:$\end{turn}}}
\newcommand{\submeet}{\mathbin{\begin{turn}{-90}$\hspace{-0.8em}<:$\end{turn}}}

\newcommand{\cssubmeet}{\consistent{\submeet}}
\newcommand{\cssubjoin}{\consistent{\subjoin}}

\newcommand{\ltop}{\textsf{H}}
\newcommand{\lbot}{\textsf{L}}
\newcommand{\lx}{\ell} % label
\newcommand{\ul}{\?}% unknown label
\newcommand{\clx}{{\tilde{\lx}}} % gradual label
\newcommand{\cll}{\collecting{\lx}} % collecting label
\newcommand{\cS}{{\consistent{S}}} % gradual sec type
\newcommand{\clS}{\collecting{S}}% collecting sec type
\newcommand{\subl}{\preccurlyeq}% sub on labels
\newcommand{\csubl}{\;\consistent{\subl}\;}
\newcommand{\ljoincore}{\begin{turn}{90}$\hspace{-0.2em}\prec$\end{turn}}
\newcommand{\lmeetcore}{\begin{turn}{270}$\hspace{-0.6em}\prec$\end{turn}}
\newcommand{\ljoin}{\mathbin{\ljoincore}}
\newcommand{\lmeet}{\mathbin{\lmeetcore}}
\newcommand{\cjoin}{\mathbin{\consistent{\ljoincore}}}
\newcommand{\cmeet}{\mathbin{\consistent{\lmeetcore}}}
\newcommand{\cljoin}{\cjoin}
\newcommand{\clmeet}{\cmeet}
\newcommand{\gammal}{\gamma_\lx}
\newcommand{\alphal}{\alpha_\lx}
\newcommand{\gammas}{\gamma_S}

\newcommand{\alphas}{\alpha_S}

\newcommand{\slabel}{\mathit{label}}

\newcommand{\interior}[1]{\mathcal{I}_{#1}}

\newcommand{\Isub}{\interior{<:}}
\newcommand{\trans}[1]{\circ^{#1}}
\renewcommand{\merge}[1]{\triangle^{#1}}

\newcommand{\TermT}[1]{\oblset{Term}_{#1}}

\newcommand{\EvTerm}{\oblset{EvTerm}}
\newcommand{\EvValue}{\oblset{EvValue}}
\newcommand{\SimpleValue}{\oblset{SimpleValue}}
\newcommand{\EvFrame}{\oblset{EvFrame}}
\newcommand{\TmFrame}{\oblset{TmFrame}}
\newcommand{\pr}[1]{\braket{{#1}}}
\newcommand{\cast}[2]{\evcast{\evpr{#1}}{#2}}
\newcommand{\error}{\textup{\textbf{error}}}
\newcommand{\evt}{\mathit{et}}
\newcommand{\evv}{\mathit{ev}}
\newcommand{\red}{\longmapsto}
\newcommand{\iapp}[1]{\mathbin{@^{#1}}}
\newcommand{\invcod}{\mathit{icod}}
\newcommand{\invdom}{\mathit{idom}}

\newcommand{\ie}{\emph{i.e.}\xspace}

\newcommand{\eg}{\emph{e.g.}\xspace}

\newsavebox{\cTsave}
\savebox{\cTsave}{$\cT$}
\newsavebox{\clTsave}
\savebox{\clTsave}{$\clT$}

\newcommand{\ev}{\varepsilon}
\newcommand{\evcast}[2]{#1#2}

\newcommand{\evpr}[1]{\braket{#1}}